\documentclass[twocolumns]{IEEEtran}

\usepackage{amssymb}
\usepackage{amsmath}
\usepackage{stfloats}
\usepackage{graphicx}
\usepackage{subfigure}
\usepackage{tabularx}
\usepackage{epsfig,epsf,color,balance,cite}
\usepackage[normalem]{ulem}
\usepackage{bbding}
\usepackage{multirow}

\newtheorem{cor}{Corollary}
\newtheorem{theorem}{Theorem}

\newtheorem{remark}{Remark}

\begin{document}
\title{Physical-layer Security for Indoor Visible Light Communications: Secrecy Capacity Analysis}

\author{Jin-Yuan Wang, \emph{Member, IEEE}, Cheng Liu, Jun-Bo Wang, \emph{Member, IEEE}, Yongpeng Wu, \emph{Senior Member, IEEE}, Min Lin, \emph{Member, IEEE}, and Julian Cheng, \emph{Senior Member, IEEE}
\thanks{This work was supported in part by
National Natural Science Foundation of China under Grants 61701254 and 61571115,
Natural Science Foundation of Jiangsu Province under Grant BK20170901,
Key International Cooperation Research Project under Grant 61720106003,
the open research fund of National Mobile Communications Research Laboratory, Southeast University under Grant 2017D06,
the open research fund of Key Lab of Broadband Wireless Communication and Sensor Network Technology (Nanjing University of Posts and Telecommunications), Ministry of Education under Grants JZNY201701 and JZNY201706,
and NUPTSF under Grant NY216009.}
\thanks{Jin-Yuan Wang and Min Lin are with Key Lab of Broadband Wireless Communication and Sensor Network Technology, Nanjing University of Posts and Telecommunications, Nanjing 210003, China, and also with National Mobile Communications Research Laboratory, Southeast University, Nanjing 210096, China (e-mail: jywang@njupt.edu.cn, linmin@njupt.edu.cn).}
\thanks{Cheng Liu and Jun-Bo Wang are with National Mobile Communications Research Laboratory, Southeast University, Nanjing 210096, China (e-mail: 220160880@seu.edu.cn, jbwang@seu.edu.cn).}
\thanks{Yongpeng Wu is with Department of Electronic Engineering, Shanghai Jiao Tong University, Minhang 200240, China (e-mail: yongpeng.wu@sjtu.edu.cn).}
\thanks{Julian Cheng is with School of Engineering, The University of British Columbia, Kelowna, BC, Canada V1V 1V7 (e-mail: julian.cheng@ubc.ca).}
\thanks{Corresponding author: Jun-Bo Wang (e-mail: jbwang@seu.edu.cn)}
}

\maketitle
\begin{abstract}
This paper investigates the physical-layer security for an indoor visible light communication (VLC) network consisting of a transmitter, a legitimate receiver and an eavesdropper.
Both the main channel and the wiretapping channel have non-negative inputs,
which are corrupted by additive white Gaussian noises.
Considering the illumination requirement and the physical characteristics of lighting source,
the input is also constrained in both its average and peak optical intensities.
Two scenarios are investigated: one is only with an average optical intensity constraint,
and the other is with both average and peak optical intensity constraints.
Based on information theory,
closed-form expressions of the upper and lower bounds on secrecy capacity for the two scenarios are derived.
Numerical results show that the upper and lower bounds on secrecy capacity are tight,
which validates the derived closed-form expressions.
Moreover, the asymptotic behaviors in the high signal-to-noise ratio (SNR) regime are analyzed from the theoretical aspects.
At high SNR, when only considering the average optical intensity constraint,
a small performance gap exists between the asymptotic upper and lower bounds on secrecy capacity.
When considering both average and peak optical intensity constraints,
the asymptotic upper and lower bounds on secrecy capacity coincide with each other.
These conclusions are also confirmed by numerical results.
\end{abstract}


\IEEEpeerreviewmaketitle

\section{Introduction}
\label{section1}
With the widespread use of light-emitting diodes (LEDs) for commercial lighting applications,
visible light communication (VLC) has attracted increasing attention in recent years.
Due to the combination of communication and illumination,
VLC is regarded as one of the most important wireless communication technologies for future indoor access \cite{BIB01}.

In the last decade, the point-to-point (P2P) VLC has achieved rapid development in many fields,
especially in channel modelling \cite{BIB02}, modulation \cite{BIB03}, coding \cite{BIB04}, equalization \cite{BIB05},
channel estimation \cite{BIB06}, indoor positioning \cite{BIB06_1},
channel capacity analysis \cite{BIB07,BIB08,BIB09,BIB10} and transceiver design \cite{BIB11}.
At present, the research focus of VLC is being changed from P2P communications to network aspects.
In VLC networks, data privacy is becoming a main concern for users.
Although it is propagated via the line-of-sight path,
the VLC signal is broadcasted to all users illuminated by the LEDs.
Such a broadcast feature provides convenience for data transmission,
but it also offers an opportunity for unintended users to eavesdrop the information, which imposes
a security risk to legitimate users.
Therefore, information security becomes an urgent issue to be addressed.
Traditional security schemes are performed at upper-layers of the network stack by using access control,
password protection and end-to-end encryption \cite{BIB11_1}.
The safety of traditional security schemes is built on the limited storage capacity and computational power of the eavesdroppers.
Recently, physical-layer (PHY) security,
which exploits the channel characteristics to hide information from eavesdroppers and does not rely on the upper-layer encryption,
has been proposed as an efficient supplement to traditional security schemes.

Secure transmission is important for radio frequency wireless communications (RFWC).
The PHY security was first investigated in 1949 by Shannon,
who proposed the concept of perfect secrecy over noiseless channels \cite{BIB12}.
Under the noisy channels, Wyner analyzed the secrecy capacity via the wiretap channel \cite{BIB13}.
In \cite{BIB14}, the secrecy capacity of the single-input single-output (SISO) Gaussian wiretap channel was derived.
Under the non-degraded wiretap channel, a single-letter characterization of the secrecy capacity was derived in \cite{BIB15}.
Recently, the secrecy performance analysis over the SISO scenario was extended to that over the multi-input multi-output (MIMO) scenario.
For MIMO wiretap channels with confidential messages, the authors in \cite{BIB16} analyzed the secrecy capacity region.
For artificial noisy MIMO channels, the secrecy capacity was studied in \cite{BIB17} by using the ordered eigenvalues of Wishart matrices.
The authors in \cite{BIB18} obtained the secrecy capacity for MIMO channels with finite memory.
With the help of a cooperative jammer, a lower bound of the secrecy capacity for the MIMO channels was derived in \cite{BIB19}.

Although much work has been done to investigate the secrecy capacity for RFWC,
the developed theory is not directly applicable to VLC.
The main differences between RFWC and VLC are highlighted as follows.
First, the transmit signal in RFWC can be bipolar or unipolar,
while the signal in VLC must be unipolar because the optical intensity is typically used to carry information.
Moreover, the average power in RFWC is the mean square value of a signal,
but the average power in VLC is the mean value of the signal \cite{BIB19_1}.
Also, a lower average power is usually preferred for RFWC,
but VLC has a predefined average intensity according to the dimming target,
which is not an objective function but a constraint \cite{BIB19_1}.
Therefore, the aforementioned features should be considered for practical VLC.
In \cite{BIB20}, the secrecy capacity was analyzed for direct current biased VLC,
where a uniform input distribution is used to derive the lower bound of secrecy capacity.
In \cite{BIBA}, the secrecy capacity of multiple-input single-output (MISO) VLC channels was investigated,
where the input distribution was chosen as a truncated generalized normal (TGN) distribution.
To obtain the optimal and robust beamforming, the authors in \cite{BIBB} also employed the TGN distribution for the input.
Due to the constraints of the input signal in VLC, the uniform and TGN input distributions are generally not optimal \cite{BIB07,BIB09}.
By using the variational method, an improved input distribution can be obtained \cite{BIB07,BIB09}.
By employing a discrete input distribution, the authors in \cite{BIBC} derived
an upper bound on the secrecy capacity for the MISO VLC channels. However, the theoretical expression of the secrecy capacity is not obtained.
In \cite{BIB20_add1}, the secrecy outage probability (SOP) was analyzed for a hybrid VLC-RFWC system with energy harvesting.
In \cite{BIB20_add2}, the SOP and the average secrecy capacity were discussed for VLC with spatially random terminals.
Note that the dimming requirement for indoor VLC was not considered in \cite{BIB20_add1,BIB20_add2}.
In our previous work \cite{BIB20_add3}, three lower bounds on the secrecy capacity were obtained. However, no upper bound has been obtained.
Moreover, the peak optical intensity constraint is not considered.
To the best of the authors' knowledge, the secrecy capacity for VLC has not been systematically investigated.

In this paper, the secrecy capacity for an indoor VLC system with a transmitter, a legitimate receiver, and an eavesdropper is investigated.
The main contributions are summarized as follows:
\begin{enumerate}
  \item The secrecy capacity for indoor VLC with only an average optical intensity constraint is analyzed.
        By using the existing channel capacity results, the entropy-power inequality (EPI) and the variational method,
        two lower bounds on secrecy capacity are obtained.
        Applying the dual expression of the secrecy capacity, the upper bound on the secrecy capacity is obtained.
        Numerical results validate the derived closed-form expressions.
  \item The secrecy capacity for indoor VLC with both average and peak optical intensity constraints is investigated.
        In practical VLC, the peak optical intensity of the LED is also limited.
        By adding a peak optical intensity constraint on the input signal, the lower and upper bounds on the secrecy capacity are further derived, which are in closed forms. The accuracy of the derived closed-form expressions are confirmed by numerical results.
  \item The asymptotic behaviors at high signal-to-noise ratio (SNR) are analyzed.
        Through theoretical analysis, it is shown that the asymptotic lower and upper bounds do not coincide but with a small gap when only considering the average optical intensity constraint.
        When considering both average and peak optical intensity constraints, the asymptotic lower and upper bounds coincide, and thus the secrecy capacity can be obtained precisely.
\end{enumerate}

The reminder of this paper is organized as follows.
Section \ref{section2} describes the system model.
With different constraints, the secrecy capacity bounds and the asymptotic behavior are analyzed in Section \ref{section3} and Section \ref{section4}, respectively.
Numerical results are presented in Section \ref{section5}.
Conclusions and future directions are presented in Section \ref{section6}.

\emph{Notations}:
Throughout this paper, italicized symbols denote scalar values;
$N\left( {0,{\sigma ^2}} \right)$ stands for a Gaussian distribution with zero mean and variance ${\sigma ^2}$;
$E( \cdot )$ denotes the expectation operator; ${\rm var}(\cdot)$ denotes the variance of a variable;
${f_X}(x)$ denotes the probability density function (PDF) of $X$.
We use ${\cal H}( \cdot )$ for the entropy, $D\left( { \cdot \left\|  \cdot  \right.} \right)$ for the relative entropy, and
$I( \cdot;\cdot )$ for the mutual information.
We use $\ln ( \cdot )$ for the natural logarithm and ${\cal Q}( \cdot )$ for the Gaussian Q-function.

\begin{figure}
\centering
\includegraphics[width=7cm]{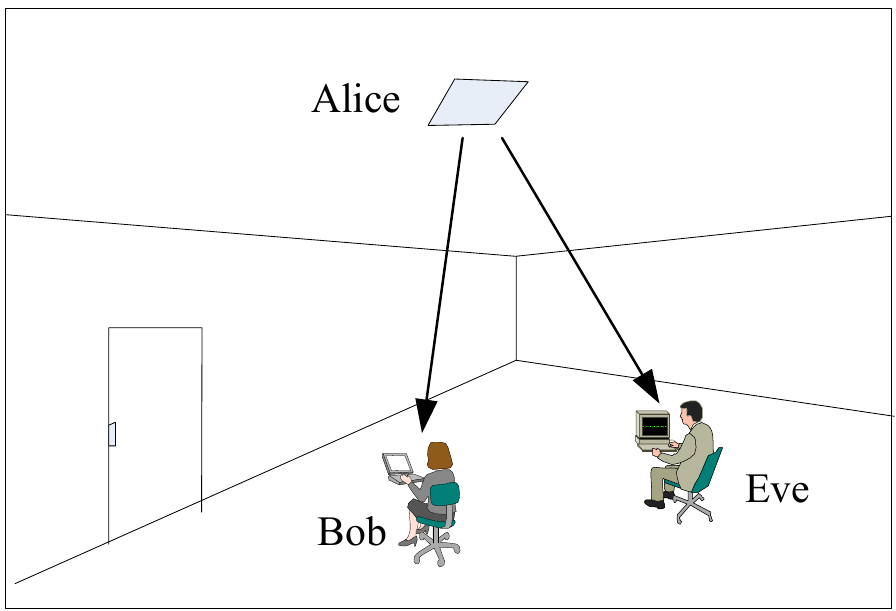}
\caption{An VLC network with one transmitter, one legitimate receiver and one eavesdropper.}
\label{fig1}
\end{figure}

\section{System Model}
\label{section2}
As shown in Fig. \ref{fig1}, we consider an indoor VLC network consisting of a transmitter (i.e., Alice), a legitimate receiver (i.e., Bob),
and an eavesdropper (i.e., Eve).
Alice is deployed on the ceiling, while Bob and Eve are placed on the floor.
When Alice transmits data bits to Bob, Eve as a passive eavesdropper can also receive the signals intended for Bob.
In the network, Alice is equipped with a single LED to transmit optical intensity signals,
while Bob and Eve are equipped with one photodiode (PD) individually to perform the optical-to-electrical conversions.
The received signals at Bob and Eve can be expressed, respectively, as
\begin{equation}
\left\{ \begin{array}{l}
{Y_{\rm{B}}} = {H_{\rm{B}}}X + {Z_{\rm{B}}}\\
{Y_{\rm{E}}} = {H_{\rm{E}}}X + {Z_{\rm{E}}}
\end{array} \right.,
\label{eq1}
\end{equation}
where $X$ is the transmit optical intensity signal; ${H_{\rm{B}}}$ and ${H_{\rm{E}}}$ denote the channel gains of the main channel and the eavesdropping channel, respectively;
${Z_{\rm{B}}} \sim N(0,\sigma _{\rm{B}}^2)$ and ${Z_{\rm{E}}} \sim N(0,\sigma _{\rm{E}}^2)$ stand for the additive white Gaussian noises at Bob and Eve,
where $\sigma _{\rm{B}}^2$ and $\sigma _{\rm{E}}^2$ denote the variances of the noises at Bob and Eve, respectively.

Because the intensity modulation and direct detection is employed for VLC,
$X$ is restricted to be nonnegative such that
\begin{equation}
X \ge 0.
\label{eq2}
\end{equation}

In practical VLC systems, the peak optical intensity of the LED is also limited.
Therefore, the peak optical intensity constraint is given by
\begin{equation}
X \le A,
\label{eq3}
\end{equation}
where $A$ is the peak optical intensity of the LED.

For a practical LED, its average optical intensity is constrained by the nominal optical intensity.
In order to satisfy the illumination requirements in VLC,
the average optical intensity cannot change with time.
Mathematically, the average optical intensity constraint is given by
\begin{equation}
E(X) = \xi P,
\label{eq4}
\end{equation}
where $\xi  \in (0,1]$ is the dimming target,
$P \in (0,A]$ is the nominal optical intensity of the LED.

In indoor VLC, the channel gain ${H_k}$ ($k = {\rm{B}}\;{\rm{or}}\;{\rm{E}}$) can be expressed as \cite{BIB20_1}
\begin{equation}
{H_k} \!\!=\!\! \left\{ \begin{array}{l}
\!\!\!\!\!\frac{{(m + 1){A_r}}}{{2\pi D_k^2}}{T_s}g{\cos ^m}({\varphi _k})\cos ({\psi _k}),\;{\rm{if}}\;0 \le {\psi _k} \le \Psi \\
\;\;\;\;\;\;\;\;\;\;\;\;\;\;\;\;\;\;\;0,\;\;\;\;\;\;\;\;\;\;\;\;\;\;\;\;\;\;\;\;{\rm{if}}\;{\psi _k} \ge \Psi
\end{array} \right.\!\!\!,
\label{eq5}
\end{equation}
where $m$ is the order of the Lambertian emission;
${A_r}$ is the physical area of the PD;
${T_s}$ and $g$ are the optical filter gain and the concentrator gain of the PD.
$\Psi$ is the field of view (FOV) of the PD;
${D_k}$, ${\varphi _k}$ and ${\psi _k}$ are respectively the distance,
the irradiance angle and the incidence angle from Alice to Bob ($k = {\rm{B}}$) or Eve ($k = {\rm{E}}$),
as shown in Fig. \ref{fig2}.
Obviously, when the positions of Alice, Bob and Eve are fixed, the channel gains $H_{\rm B}$ and $H_{\rm E}$ are constants.

\begin{figure}
\centering
\includegraphics[width=7.5cm]{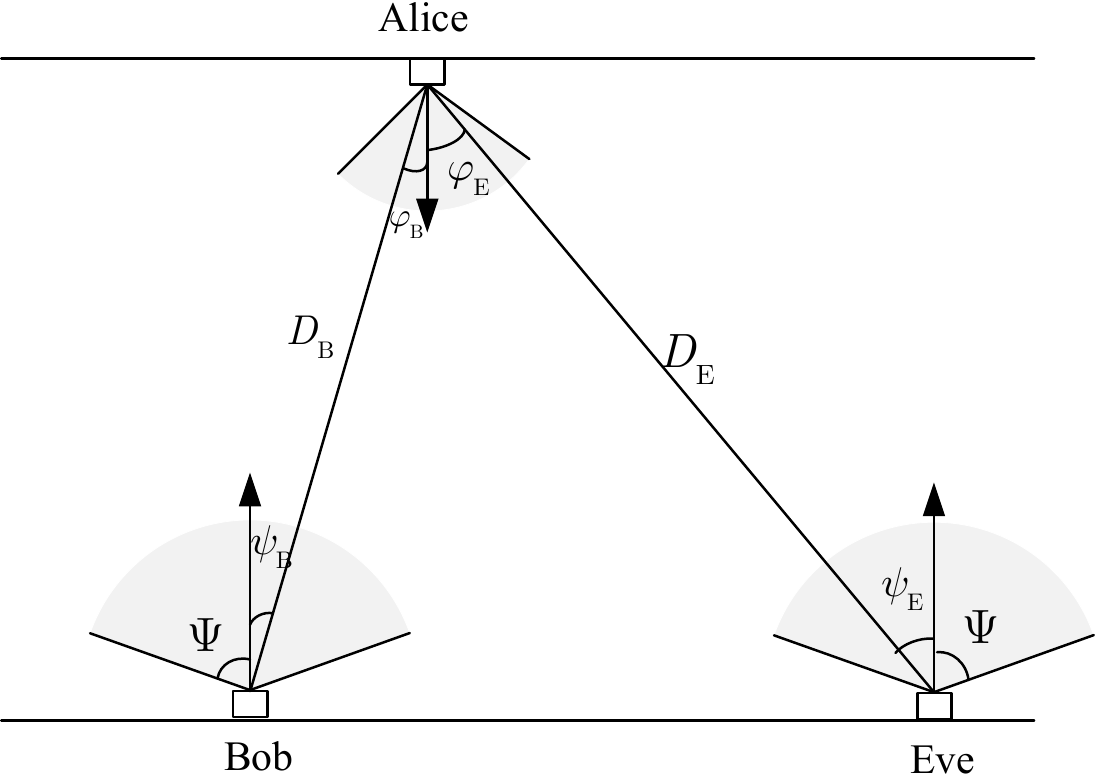}
\caption{The distance and angle relationship among Alice, Bob and Eve.}
\label{fig2}
\end{figure}

\section{Secrecy Capacity for VLC with Only an Average Optical Intensity Constraint}
\label{section3}
In this section, we focus on the VLC with only an average optical intensity constraint.
Therefore, we only consider the constraints (\ref{eq2}) and (\ref{eq4}).
If the main channel is worse than the eavesdropping channel (i.e., $H_{\rm B}/\sigma_{\rm B} < H_{\rm E}/\sigma_{\rm E}$), then the main channel is stochastically degraded with respect to the eavesdropping channel, and the secrecy capacity is essentially zero. Alternatively, if $H_{\rm B}/\sigma_{\rm B} \geq H_{\rm E}/\sigma_{\rm E}$, the secrecy capacity\footnote{In this paper, the natural logarithms are employed, and thus the secrecy capacity is in nats/transmission.} can be expressed as \cite{BIB20,BIB20_add4}
\begin{eqnarray}
&&C_{\rm s} = \mathop {\max }\limits_{{f_X}(x)} \left[ {I(X;{Y_{\rm{B}}}) - I(X;{Y_{\rm{E}}})} \right] \nonumber\\
{\rm{s.t.}}&& \int_0^\infty  {{f_X}(x){\rm d}x}  = 1\nonumber\\
&&E(X) =  \int_0^\infty  {x{f_X}(x){\rm d}x}  = \xi P,
\label{eq6}
\end{eqnarray}
where $C_{\rm s}$ denotes the secrecy capacity, and ${f_X}(x)$ denotes the PDF of the input signal.
In this section, the lower and upper bounds of the secrecy capacity will be derived.
Moreover, the asymptotic behavior of the secrecy capacity is also analyzed.

\subsection{Lower Bound on Secrecy Capacity}
\label{section3_1}
By using the channel capacity bounds derived in \cite{BIB22},
a lower bound on the secrecy capacity is derived as (\ref{eq7}) in the following theorem.
Moreover, by using the EPI and the variational method,
another lower bound on the secrecy capacity can be derived as (\ref{eq8}) in the following theorem.

\begin{theorem}
The secrecy capacity for the channel in (\ref{eq1}) with constraints (\ref{eq2}) and (\ref{eq4}) is lower-bounded by each of the following two bounds (\ref{eq7}) and (\ref{eq8}), i.e.,
\begin{eqnarray}
C_{\rm s} \!\!&\ge&\!\! \ln \left[ {\frac{{{\sigma _{\rm{E}}} {\sqrt {2\pi e\left( {1 + \frac{{H_{\rm{B}}^2{\xi ^2}{P^2}e}}{{2\pi \sigma _{\rm{B}}^2}}} \right)} } }}{{\beta {e^{ - \frac{{{\delta ^2}}}{{2\sigma _{\rm{E}}^2}}}} + \sqrt {2\pi } {\sigma _{\rm{E}}}{\cal Q}\left( {\frac{\delta }{{{\sigma _{\rm{E}}}}}} \right)}}} \right] - \frac{1}{2}{\cal Q}\left( {\frac{\delta }{{{\sigma _{\rm{E}}}}}} \right) \nonumber\\
&-& \!\!\frac{\delta }{{2\sqrt {2\pi } {\sigma _{\rm{E}}}}}{e^{ - \frac{{{\delta ^2}}}{{2\sigma _{\rm{E}}^2}}}} - \frac{{{\delta ^2}}}{{2\sigma _{\rm{E}}^2}}{\cal Q}\left( { - \frac{{\delta  + {H_{\rm{E}}}\xi P}}{{{\sigma _{\rm{E}}}}}} \right)\nonumber \\
 &-&\!\!  \frac{{\delta  + {H_{\rm{E}}}\xi P}}{\beta } - \frac{{{\sigma _{\rm{E}}}}}{{\sqrt {2\pi } \beta }}{e^{ - \frac{{{\delta ^2}}}{{2\sigma _{\rm{E}}^2}}}},
 \label{eq7}
\end{eqnarray}
and
\begin{equation}
C_{\rm s} \ge \frac{1}{2}\ln \left( {\frac{{\sigma _{\rm{E}}^2}}{{2\pi \sigma _{\rm{B}}^2}} \cdot \frac{{e{\xi ^2}{P^2}H_{\rm{B}}^2 + 2\pi \sigma _{\rm{B}}^2}}{{H_{\rm{E}}^2{\xi ^2}{P^2} + \sigma _{\rm{E}}^2}}} \right),
 \label{eq8}
\end{equation}
where $\beta >0$ and $\delta \ge 0$ in (\ref{eq7}) are free parameters.
Suboptimal but useful choices for $\beta$ and $\delta$ are given, respectively, by (\ref{eq9}) and (\ref{eq10}) as shown at the top of the next page.
\begin{table*}\normalsize
\begin{eqnarray}
\beta \!=\! \frac{1}{2}\!\!\left(\! {\delta  \!+\! {H_{\rm{E}}}\xi P \!+\! \frac{{{\sigma _{\rm{E}}}}}{{\sqrt {2\pi } }}{e^{ - \frac{{{\delta ^2}}}{{2\sigma _{\rm{E}}^2}}}}} \!\right)\!+\! \frac{1}{2}\sqrt {\!{{\left( {\delta \!+\! {H_{\rm{E}}}\xi P \!+\! \frac{{{\sigma _{\rm{E}}}}}{{\sqrt {2\pi } }}{e^{ - \frac{{{\delta ^2}}}{{2\sigma _{\rm{E}}^2}}}}} \right)}^2} \!+\! 4\!\left( {\delta  \!+\! {H_{\rm{E}}}\xi P \!+\! \frac{{{\sigma _{\rm{E}}}}}{{\sqrt {2\pi } }}{e^{ - \frac{{{\delta ^2}}}{{2\sigma _{\rm{E}}^2}}}}} \right)\!\!\sqrt {2\pi } {\sigma _{\rm{E}}}{e^{\frac{{{\delta ^2}}}{{2\sigma _{\rm{E}}^2}}}}{\cal Q}\left( {\frac{\delta }{{{\sigma _{\rm{E}}}}}} \right)},
  \label{eq9}
\end{eqnarray}

\begin{equation}
\delta  = {\sigma _{\rm{E}}}\ln \left( {1 + \frac{{{H_{\rm{E}}}\xi P}}{{{\sigma _{\rm{E}}}}}} \right).
 \label{eq10}
\end{equation}
\hrulefill
\end{table*}
\label{them1}
\end{theorem}

\begin{proof}
See Appendix \ref{appa} and Appendix \ref{appb}.
\end{proof}

\begin{remark}
When $\sqrt{e} H_{\rm B}> H_{\rm E}$, it is obvious that the lower bound (\ref{eq8}) is a monotonically increasing function with respect to the dimming target $\xi$.
In this case, the performance of (\ref{eq8}) can be improved by increasing the value of $\xi$.
Moreover, it is challenging to analyze the monotonicity of the lower bound (\ref{eq7}) with respect to $\xi$,
which will be given by using a numerical method in Section \ref{section5}.
\end{remark}

\subsection{Upper Bound on Secrecy Capacity}
\label{section3_2}
To obtain the upper bound on the secrecy capacity,
the dual expression of the secrecy capacity is employed \cite{BIB21}.
For an arbitrary conditional PDF ${g_{{Y_{\rm{B}}}|{Y_{\rm E}}}}({y_{\rm{B}}}|{y_{\rm{E}}})$,
the following equality always holds \cite{BIB21_1}
\begin{eqnarray}
&&\!\!\!\!\!\!\!\!\!\!I(X;{Y_{\rm{B}}}|{Y_{\rm{E}}}) \!+\! {E_{X{Y_{\rm{E}}}}}\!\left\{ {D\left( {\left. {{f_{{Y_{\rm{B}}}|{Y_{\rm E}}}}({y_{\rm{B}}}|{Y_{\rm{E}}})} \right\|{g_{{Y_{\rm{B}}}|{Y_{\rm E}}}}({y_{\rm{B}}}|{Y_{\rm{E}}})} \right)} \right\} \nonumber \\
&&\!\!\!\!\!\!\!\!\!\! = {E_{X{Y_{\rm{E}}}}}\left\{ {D\left( {\left. {{f_{{Y_{\rm{B}}}|X{Y_{\rm E}}}}({y_{\rm{B}}}|X,{Y_{\rm{E}}})} \right\|{g_{{Y_{\rm{B}}}|{Y_{\rm E}}}}({y_{\rm{B}}}|{Y_{\rm{E}}})} \right)} \right\}.
 \label{eq11}
\end{eqnarray}

According to the non-negativity property of the relative entropy, we have
\begin{eqnarray}
&&\!\!\!\!\!\!\!\!\!\!\!\!\!\!I(X;{Y_{\rm{B}}}|{Y_{\rm{E}}})\nonumber\\
&&\!\!\!\!\!\!\!\!\le {E_{X{Y_{\rm{E}}}}}\!\!\left\{ {D\!\left( {\left. {{f_{{Y_{\rm{B}}}|X{Y_{\rm E}}}}({y_{\rm{B}}}|X,{Y_{\rm{E}}})} \right\|{g_{{Y_{\rm{B}}}|{Y_{\rm E}}}}({y_{\rm{B}}}|{Y_{\rm{E}}})} \right)}\! \right\}\!\!.
 \label{eq13}
\end{eqnarray}
Note that selecting any ${g_{{Y_{\rm{B}}}|{Y_{\rm E}}}}({y_{\rm{B}}}|{y_{\rm{E}}})$ will result in an upper bound of $I(X;{Y_{\rm{B}}}|{Y_{\rm{E}}})$. Therefore, we have
\begin{eqnarray}
&&\!\!\!\!\!\!\!\!\!\!\!\!\!\!\!\!\!\!\!I(X;{Y_{\rm{B}}}|{Y_{\rm{E}}})= \nonumber\\
&&\!\!\!\!\!\!\!\!\!\!\!\!\!\!\!\!\!\!\!\!\! \mathop {\min }\limits_{{g_{{Y_{\rm{B}}}|{Y_{\rm E}}}}\!({y_{\rm{B}}}|{Y_{\rm{E}}})} \!\!\!{E_{X\!{Y_{\rm{E}}}}}\!\!\left\{\! {D\!\!\left(\! {\left. {{f_{{Y_{\rm{B}}}|X\!{Y_{\rm E}}}}({y_{\rm{B}}}|X,\!{Y_{\rm{E}}})} \right\|\!{g_{{Y_{\rm{B}}}|{Y_{\rm E}}}}\!({y_{\rm{B}}}|{Y_{\rm{E}}})}\! \right)}\! \right\}\!\!.
 \label{eq14}
\end{eqnarray}
Note that, to achieve the secrecy capacity, there exists a unique input PDF ${f_{{X^*}}}(x)$ that maximizes $I(X;{Y_{\rm{B}}}|{Y_{\rm{E}}})$ subject to the constraints in (\ref{eq6}). Therefore, we have
\begin{eqnarray}
C_{\rm s} &=& \mathop {\max }\limits_{{f_X}(x)} I(X;Y_{\rm B}| Y_{\rm E})\nonumber\\
 &=& I(X^*;Y_{\rm B}| Y_{\rm E}),
 \label{eq15}
\end{eqnarray}
where ${X^*}$ denotes the optimal input, and its corresponding PDF is ${f_{{X^*}}}(x)$.

According to (\ref{eq14}) and (\ref{eq15}), we have
\begin{equation}
C_{\rm s} \!\!\le\!\! {E_{{X^*}{Y_{\rm{E}}}}}\!\left\{\! {D\!\left( {\left. {{f_{{Y_{\rm{B}}}|X{Y_{\rm E}}}}({y_{\rm{B}}}|X,{Y_{\rm{E}}})} \right\|{g_{{Y_{\rm{B}}}|{Y_{\rm E}}}}({y_{\rm{B}}}|{Y_{\rm{E}}})} \right)}\! \right\}\!.
 \label{eq16}
\end{equation}
It can be seen from (\ref{eq16}) that selecting any ${g_{{Y_{\rm{B}}}|{Y_{\rm E}}}}({y_{\rm{B}}}|{y_{\rm{E}}})$ will lead to an upper bound of the secrecy capacity.
To obtain a good upper bound, a clever choice of ${g_{{Y_{\rm{B}}}|{Y_{\rm E}}}}({y_{\rm{B}}}|{y_{\rm{E}}})$ should be found.
By using the principle of dual expression of the secrecy capacity,
an upper bound of the secrecy capacity in (\ref{eq6}) is derived in the following theorem.

\begin{theorem}
The secrecy capacity for the channel in (\ref{eq1}) with constraints (\ref{eq2}) and (\ref{eq4}) is upper-bounded by
\begin{equation}
C_{\rm s} \!\!\le\!\!\! \left\{ \begin{array}{l}\!\!\!\!\!
\ln\!\! \left[\! {\frac{{4e\left(\! {\sqrt {\frac{1}{{2\pi }}} {\sigma _{\rm{B}}} \!+\! \frac{{{H_{\rm{B}}}\xi P}}{2}} \!\right)}}{{\sqrt {2\pi e\sigma _{\rm{B}}^2\left(\! {1 \!+\! \frac{{H_{\rm{E}}^2\sigma _{\rm{B}}^2}}{{H_{\rm{B}}^2\sigma _{\rm{E}}^2}}} \!\right)} }}}\! \right]\!,{\kern 1pt} {\rm{if}}{\kern 1pt} \sqrt {\!\!\frac{{\frac{{H_{\rm{E}}^2}}{{H_{\rm{B}}^2}}\sigma _{\rm{B}}^2 \!+\! \sigma _{\rm{E}}^2}}{{2\pi }}} \! \ge\!\! \frac{{{H_{\rm{E}}}}}{{{H_{\rm{B}}}}}\!\!\left(\! {\frac{{{\sigma _{\rm{B}}}}}{{\sqrt {2\pi } }} \!+\! \frac{{{H_{\rm{B}}}\xi P}}{2}}\! \right)\\
\!\!\!\! \ln \left( {\frac{{2\sqrt e {H_{\rm{B}}}{\sigma _{\rm{E}}}}}{{\pi {H_{\rm{E}}}{\sigma _{\rm{B}}}}}} \right),{\kern 1pt} {\rm{if}}{\kern 1pt} \sqrt {\!\!\frac{{\frac{{H_{\rm{E}}^2}}{{H_{\rm{B}}^2}}\sigma _{\rm{B}}^2 \!+\! \sigma _{\rm{E}}^2}}{{2\pi }}} \! \le\!\! \frac{{{H_{\rm{E}}}}}{{{H_{\rm{B}}}}}\!\!\left(\! {\frac{{{\sigma _{\rm{B}}}}}{{\sqrt {2\pi } }} \!+\! \frac{{{H_{\rm{B}}}\xi P}}{2}} \!\right).
\end{array} \right.
 \label{eq17}
\end{equation}
\label{them2}
\end{theorem}

\begin{proof}
See Appendix \ref{appc}.
\end{proof}

\begin{remark}
When $\sqrt {{{(H_{\rm{E}}^2\sigma _{\rm{B}}^2/H_{\rm{B}}^2 + \sigma _{\rm{E}}^2)} /{{\rm{(2}}\pi )}}}  \ge \frac{{{H_{\rm{E}}}}}{{{H_{\rm{B}}}}}\left( {\frac{{{\sigma _{\rm{B}}}}}{{\sqrt {2\pi } }} + \frac{{{H_{\rm{B}}}\xi P}}{2}} \right)$, it is straightforward to show that (\ref{eq17}) is a monotonically increasing function with respect to the dimming target $\xi$.
In this case, the larger the dimming target is, the larger the upper bound on secrecy capacity becomes.
When $\sqrt {{{(H_{\rm{E}}^2\sigma _{\rm{B}}^2/H_{\rm{B}}^2 + \sigma _{\rm{E}}^2)} / {{\rm{(2}}\pi )}}}  \leq \frac{{{H_{\rm{E}}}}}{{{H_{\rm{B}}}}}\left( {\frac{{{\sigma _{\rm{B}}}}}{{\sqrt {2\pi } }} + \frac{{{H_{\rm{B}}}\xi P}}{2}} \right)$, the upper bound (\ref{eq17}) is independent of the dimming target.
\end{remark}

\subsection{Asymptotic Behavior Analysis}
\label{section3_3}
In indoor VLC environment, typical illumination requirement leads to a large transmit optical intensity,
which can offer a high SNR at the receiver \cite{BIB21_2}.
Therefore, we are more interested in the behavior of the VLC system in the high SNR regime.
By analyzing \emph{Theorem \ref{them1}} and \emph{Theorem \ref{them2}},
the asymptotic behavior of the secrecy capacity bounds at high SNR is derived in the following corollary.

\begin{cor}
For the channel in (\ref{eq1}) with constraints (\ref{eq2}) and (\ref{eq4}), the asymptotic behavior of the secrecy capacity bounds at high SNR is given by
\begin{equation}
\ln \left( {\frac{{{H_{\rm{B}}}{\sigma _{\rm{E}}}}}{{{H_{\rm{E}}}{\sigma _{\rm{B}}}}}} \right) \le \mathop {\lim }\limits_{P \to \infty } C_{\rm s} \le \ln \left( {\frac{{2\sqrt e }}{\pi }} \right) + \ln \left( {\frac{{{H_{\rm{B}}}{\sigma _{\rm{E}}}}}{{{H_{\rm{E}}}{\sigma _{\rm{B}}}}}} \right).
 \label{eq17_1}
\end{equation}
\label{cor1}
\end{cor}

\begin{proof}
See Appendix \ref{appc1}.
\end{proof}

\begin{remark}
From \emph{Corollary \ref{cor1}}, it can be found that the asymptotic upper and lower bounds on secrecy capacity do not coincide in the sense that their difference equals $\ln (2\sqrt e /\pi ) \approx 0.048$ nats/transmission instead of zero.
Although a performance gap between the asymptotic upper and lower bounds exists, the difference is so small and thus it can be ignored.
\end{remark}

\section{Secrecy Capacity for VLC with Both Average and Peak Optical Intensity Constraints}
\label{section4}
As is well known, the peak optical intensity of an LED is also limited.
In this section, we consider the constraints (\ref{eq2}), (\ref{eq3}) and (\ref{eq4}).
Similarly, if $H_{\rm B}/\sigma_{\rm B} < H_{\rm E}/\sigma_{\rm E}$, the secrecy capacity is zero. If $H_{\rm B}/\sigma_{\rm B} \geq H_{\rm E}/\sigma_{\rm E}$, the secrecy capacity can be expressed as \cite{BIB20,BIB20_add4}
\begin{eqnarray}
&&C_{\rm s} = \mathop {\max }\limits_{{f_X}(x)} \left[ {I(X;{Y_{\rm{B}}}) - I(X;{Y_{\rm{E}}})} \right] \nonumber\\
{\rm{s.t.}} && \int_0^A {{f_X}(x) {\rm d}x}  = 1 \nonumber \\
&& \int_0^A {x{f_X}(x) {\rm d}x}  = \xi P.
 \label{eq18}
\end{eqnarray}
In this section, the bounds of the secrecy capacity and their asymptotic behaviors for this case will be derived.

\subsection{Lower Bound on Secrecy Capacity}
\label{section4_1}
Let $\alpha  = \xi P/A$ denote the average to peak optical intensity ratio (APOIR). Two lower bounds on the secrecy capacity can be derived by analyzing (\ref{eq18}), and these lower bounds are shown in the following theorem.

\begin{theorem}
The secrecy capacity for the channel in (\ref{eq1}) with constraints (\ref{eq2}), (\ref{eq3}) and (\ref{eq4}) can be lower-bounded by each of the following two bounds (\ref{eq19}) and (\ref{eq20}), i.e.,
\begin{equation}
C_{\rm s} \ge \left\{ \begin{array}{l}
{C_{{\rm s},1}},\;0 < \alpha  < 0.5\\
{C_{{\rm s},2}},0.5 \le \alpha  \le 1
\end{array} \right.\!\!\!,
 \label{eq19}
\end{equation}
and
\begin{equation}
C_{\rm s} \!\ge\! \left\{ \begin{array}{l}\!\!\!\!
\frac{1}{2}\ln\! \left[ {\frac{{3\sigma _{\rm{E}}^2\left( {H_{\rm{B}}^2{A^2} + 2\pi e\sigma _{\rm{B}}^2} \right)}}{{2\pi e\sigma _{\rm{B}}^2\left( {H_{\rm{E}}^2{\xi ^2}{P^2} + 3\sigma _{\rm{E}}^2} \right)}}} \right],\;\;\alpha  = 0.5\\
\!\!\!\! \frac{1}{2}\ln\! \left\{\! {\frac{{\sigma _{\rm{E}}^2\left[ {H_{\rm{B}}^2{e^{ - 2c\xi P}}{{\left( {\frac{{{e^{cA}} - 1}}{c}} \right)}^2} + 2\pi e\sigma _{\rm{B}}^2} \right]}}{{2\pi e\sigma _{\rm{B}}^2\left[ {\frac{{H_{\rm{E}}^2A(cA \!-\! 2)}}{{c(1 \!-\! {e^{ - cA}})}} \!+\! \frac{{{\rm{2}}H_{\rm{E}}^2}}{{{c^2}}} \!-\! H_{\rm{E}}^2{\xi ^2}{P^2} \!+\! \sigma _{\rm{E}}^2} \right]}}} \!\right\},\;\\
\;\;\;\;\;\;\;\;\;\;\;\;\;\;\;\;\;\alpha  \ne 0.5\;{\rm{and}}\;\alpha  \in (0,1]
\end{array} \right.\!\!\!\!,
 \label{eq20}
\end{equation}
where $C_{{\rm s},1}$ and $C_{{\rm s},2}$ in (\ref{eq19}) can be written, respectively, as (\ref{eq21}) and (\ref{eq22}) as shown at the top of the next page,
\begin{table*}\normalsize
\begin{eqnarray}
C_{{\rm s},1} \!\!\!&=&\!\!\!\!\! \frac{1}{2}\ln \left[ {1 + H_{\rm{B}}^2{A^2}\frac{{{e^{2\alpha \tilde \mu }}{{(1 - {e^{ - \tilde \mu }})}^2}}}{{2\pi e\sigma _{\rm{B}}^2{{\tilde \mu }^2}}}} \right] - {\cal Q}\left( {\frac{\delta }{{{\sigma _{\rm{E}}}}}} \right) - \frac{\delta }{{\sqrt {2\pi } {\sigma _{\rm{E}}}}}{e^{ - \frac{{{\delta ^2}}}{{2\sigma _{\rm{E}}^2}}}} + \frac{1}{2}  - \frac{{\mu {\sigma _{\rm{E}}}}}{{{H_{\rm{E}}}A\sqrt {2\pi } }}\left( {{e^{ - \frac{{{\delta ^2}}}{{2\sigma _{\rm{E}}^2}}}} - {e^{ - \frac{{{{({H_{\rm{E}}}A + \delta )}^2}}}{{2\sigma _{\rm{E}}^2}}}}} \right)\nonumber \\
 &-&\!\!\!\!\!\! \left[ {{\cal Q}\!\left(\! { - \frac{{\delta  \!+\! {H_{\rm{E}}}\alpha A}}{{{\sigma _{\rm{E}}}}}}\! \right) \!-\! {\cal Q}\!\left(\! {\frac{{\delta  \!+\! (1 \!-\! \alpha ){H_{\rm{E}}}A}}{{{\sigma _{\rm{E}}}}}}\! \right)}\! \right]\ln\! \left[\! {\frac{{{H_{\rm{E}}}A}}{{\sqrt {2\pi } {\sigma _{\rm{E}}}\mu }}.\frac{{{e^{\frac{{\mu \delta }}{{{H_{\rm{E}}}A}}}} \!-\! {e^{ - \mu \left( {1 + \frac{\delta }{{{H_{\rm{E}}}A}}} \right)}}}}{{\left( {1 \!-\! 2{\cal Q}\left( {\frac{\delta }{{{\sigma _{\rm{E}}}}}} \right)} \right)}}}\! \right]\!-\! \mu \alpha \!\left[\! {1 \!-\! 2{\cal Q}\!\left( {\frac{{\delta  \!+\! \frac{{{H_{\rm{E}}}A}}{2}}}{{{\sigma _{\rm{E}}}}}} \right)}\! \right],
 \label{eq21}
\end{eqnarray}

\begin{eqnarray}
C_{{\rm s},2}\!=\! \frac{1}{2}\ln \left( {1 \!+\! \frac{{H_{\rm{B}}^2{A^2}}}{{2\pi e\sigma _{\rm{B}}^2}}} \right) \!-\! \left[ {1 \!-\! 2{\cal Q}\left( {\frac{{\delta  \!+\! \frac{{{H_{\rm{E}}}A}}{2}}}{{{\sigma _{\rm{E}}}}}} \right)} \right]\ln \left[ {\frac{{{H_{\rm{E}}}A + 2\delta }}{{\sqrt {2\pi } {\sigma _{\rm{E}}}\left( {1 \!-\! 2{\cal Q}\left( {\frac{\delta }{{{\sigma _{\rm{E}}}}}} \right)} \right)}}} \right]
 \!-\! {\cal Q}\left( {\frac{\delta }{{{\sigma _{\rm{E}}}}}} \right) \!-\! \frac{\delta }{{\sqrt {2\pi } {\sigma _{\rm{E}}}}}{e^{ - \frac{{{\delta ^2}}}{{2\sigma _{\rm{E}}^2}}}} \!+\! \frac{1}{2},
\label{eq22}
\end{eqnarray}
\hrulefill
\end{table*}
where $\tilde \mu $ in (\ref{eq21}) is the solution to the following equation
\begin{equation}
\alpha  = \frac{1}{{\tilde \mu }} - \frac{{{e^{ - \tilde \mu }}}}{{1 - {e^{ - \tilde \mu }}}}.
 \label{eq23}
\end{equation}
In (\ref{eq21}) and (\ref{eq22}), $\mu >0$ and $\delta >0$ are free parameters. Suboptimal but useful choices for $\mu$ and $\delta$ are given by
\begin{equation}
\left\{ \begin{array}{l}
\delta  = {\sigma _{\rm{E}}}\ln \left( {1 + \frac{{{H_{\rm{E}}}A}}{{{\sigma _{\rm{E}}}}}} \right)\\
\mu  = \tilde \mu \left( {1 - {e^{ - \alpha \frac{{{\delta ^2}}}{{2\sigma _{\rm{E}}^2}}}}} \right).
\end{array} \right.
 \label{eq24}
\end{equation}
Moreover, $c$ in (\ref{eq20}) is the solution to the following equation
\begin{equation}
\alpha  = \frac{1}{{1 - {e^{ - cA}}}} - \frac{1}{{cA}}.
 \label{eq25}
\end{equation}
\label{them3}
\end{theorem}

\begin{proof}
See Appendix \ref{appd} and Appendix \ref{appe}.
\end{proof}

\begin{remark}
It can be easily proved that the curve of (\ref{eq20}) is symmetric with respect to $\alpha=0.5$.
Moreover, it is challenging to analyze the relationship between lower bound (\ref{eq19}) and $\alpha$.
Alternatively, numerical results are provided in Section \ref{section5}.
\end{remark}

\subsection{Upper Bound on Secrecy Capacity}
\label{section4_2}
According to (\ref{eq16}) and referring to \emph{Theorem \ref{them2}}, we have the following theorem.

\begin{theorem}
The secrecy capacity for the channel in (\ref{eq1}) with constraints (\ref{eq2}), (\ref{eq3}) and (\ref{eq4}) is upper-bounded by
\begin{equation}
C_{\rm s} \!\le \! \frac{1}{2}\ln\! \left[\! {\frac{{\left( {\frac{{H_{\rm{E}}^{\rm{2}}}}{{H_{\rm{B}}^{\rm{2}}}}\sigma _{\rm{B}}^{\rm{2}} \!+\! \sigma _{\rm{E}}^{\rm{2}}} \right)\left( {H_{\rm{B}}^2A\xi P \!+\! \sigma _{\rm{B}}^{\rm{2}}} \right)}}{{\sigma _{\rm{B}}^2\!\left(\! {H_{\rm{E}}^2A\xi P \!+\! 2\frac{{H_{\rm{E}}^{\rm{2}}}}{{H_{\rm{B}}^{\rm{2}}}}\sigma _{\rm{B}}^{\rm{2}} \!+\! \sigma _{\rm{E}}^{\rm{2}}} \!\right)\left(\! {1 \!+\! \frac{{H_{\rm{E}}^2\sigma _{\rm{B}}^2}}{{H_{\rm{B}}^2\sigma _{\rm{E}}^2}}} \!\right)}}} \!\right].
 \label{eq26}
\end{equation}
\label{them4}
\end{theorem}

\begin{proof}
See Appendix \ref{appf}.
\end{proof}

\begin{remark}
It can be easily shown that the upper bound (\ref{eq26}) is a monotonically non-decreasing function with respect to the APOIR $\alpha$.
\label{coradd1}
\end{remark}

\subsection{Asymptotic Behavior Analysis}
\label{section4_3}
By analyzing \emph{Theorem \ref{them3}} and \emph{Theorem \ref{them4}},
the asymptotic behavior of the secrecy capacity bounds in the high SNR regime is derived in the following corollary.

\begin{cor}
For the channel in (\ref{eq1}) with constraints (\ref{eq2}), (\ref{eq3}) and (\ref{eq4}),
the asymptotic behavior of the secrecy capacity bounds at asymptotically high SNR is given by
\begin{equation}
\ln \left( {\frac{{{H_{\rm{B}}}{\sigma _{\rm{E}}}}}{{{H_{\rm{E}}}{\sigma _{\rm{B}}}}}} \right) \le \mathop {\lim }\limits_{P \to \infty } C_{\rm s} \le \ln \left( {\frac{{{H_{\rm{B}}}{\sigma _{\rm{E}}}}}{{{H_{\rm{E}}}{\sigma _{\rm{B}}}}}} \right).
\label{eq26_1}
\end{equation}
\label{cor2}
\end{cor}

\begin{proof}
See Appendix \ref{appg}.
\end{proof}

\begin{remark}
From \emph{Corollary \ref{cor2}}, it can be found that the asymptotic upper and lower bounds on secrecy capacity coincide in the sense that their gap is equal to zero.
\end{remark}

\section{Numerical Results}
\label{section5}
In this section, selected numerical examples will be provided to verify the derived expressions of secrecy capacity.
Here, a practical indoor VLC system within a $10{\rm m}\times 10{\rm m}\times 3{\rm m}$ room is considered.
Alice is installed on the ceiling, whose coordinate is $(a,b,c)$.
Bob and Eve are deployed on the floor, whose coordinates are $(d,e,f)$ and $(x,y,f)$.
To facilitate the evaluation, the noise variances of Bob and Eve are assumed to be the same and normalized to be 1,
i.e., $\sigma_{\rm B}^2=\sigma_{\rm E}^2=0$ dB \cite{BIB07}.
The other simulation parameters are listed in Table \ref{tab0}.

\begin{table}[!h]
\caption{Main simulation parameters.}
\begin{center}
\begin{tabularx}{8cm}{|p{5cm}|X|X|}\hline\hline
\centering \textbf{Parameters} &\centering \textbf{Symbols} &\centering \textbf{Values}
\tabularnewline\hline
\centering Order of the Lambertian emission &\centering $m$ &\centering 6
\tabularnewline\hline
\centering Physical area of the PD &\centering $A_r$ &\centering 1cm$^2$
\tabularnewline\hline
\centering Optical filter gain of the PD &\centering $T_s$ &\centering 1
\tabularnewline\hline
\centering Concentrator gain of the PD &\centering $g$ &\centering 3
\tabularnewline\hline
\centering FOV of the PD &\centering $\Psi$ &\centering 75$^0$
\tabularnewline\hline\hline
\end{tabularx}
\end{center}
\label{tab0}
\end{table}

\subsection{Results of VLC Only with an Average Optical Intensity Constraint}
To verify the accuracy of the lower bounds (\ref{eq7}), (\ref{eq8}) and the upper bound (\ref{eq17}), Figs. \ref{fig3}-\ref{fig4_1} and Table \ref{tab1} are provided in this subsection.

\begin{figure}
\centering
\includegraphics[width=8.5cm]{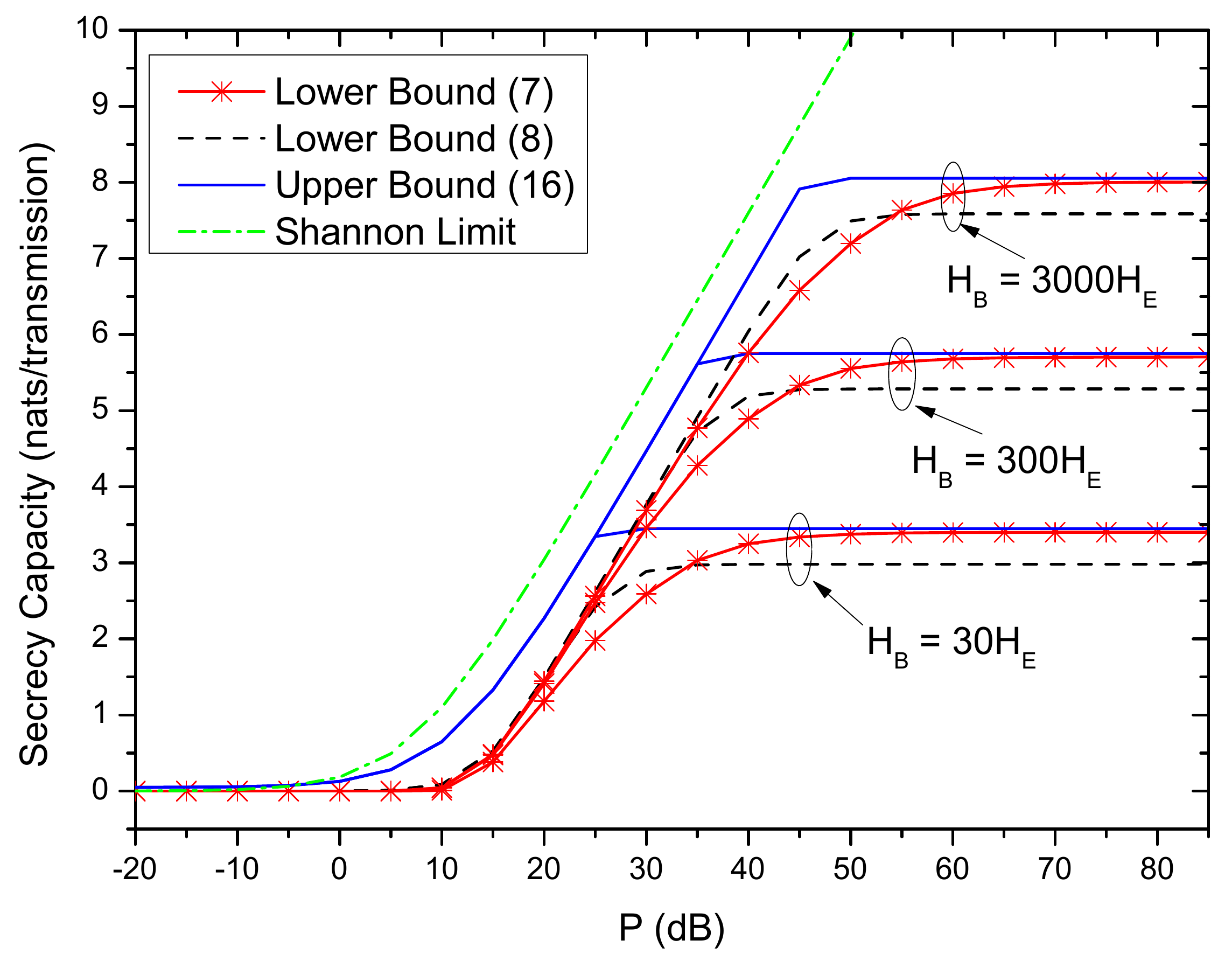}
\caption{Secrecy capacity bounds versus $P$ with different ${H_{\rm{B}}}/{H_{\rm{E}}}$ when $\xi=0.2$, $(a,b,c)=(5{\rm m}, 5{\rm m}, 3{\rm m})$ and $(d,e,f)=(5{\rm m},4.5{\rm m},0{\rm m})$.}
\label{fig3}
\end{figure}

Figure \ref{fig3} shows the secrecy capacity bounds versus $P$ with different ${H_{\rm{B}}}/{H_{\rm{E}}}$\footnote{Note that when $(x, y, f)$=(4.93m, 1.73m, 0m), ${H_{\rm{B}}}/{H_{\rm{E}}}{\rm{ = 30}}$; when $(x, y, f)$=(5.66m, 0.16m, 0m), ${H_{\rm{B}}}/{H_{\rm{E}}}{\rm{ = 300}}$; and when $(x, y, f)$=(9.7m, 9.87m, 0m), ${H_{\rm{B}}}/{H_{\rm{E}}}{\rm{ = 3000}}$.} when $\xi=0.2$, $(a,b,c)=(5{\rm m}, 5{\rm m}, 3{\rm m})$ and $(d,e,f)=(5{\rm m},4.5{\rm m},0{\rm m})$.
For comparison, the Shannon limit is also presented. Obviously, the Shannon capacity is always larger than the secrecy capacity bounds. When $P$ is small, all secrecy capacity bounds increase rapidly with the increase of $P$.
When $P$ is large, with the increase of $P$, the secrecy capacity bounds increase slowly and then tend to stable values.
Moreover, with the increase of ${H_{\rm{B}}}/{H_{\rm{E}}}$, the secrecy capacity bounds also increase.
This indicates that the larger the difference between $H_{\rm{B}}$ and $H_{\rm{E}}$ is, the better the system performance becomes.
It can also be observed that the performance of (\ref{eq8}) outperforms that of (\ref{eq7}) at low SNR.
However, at high SNR, eq. (\ref{eq7}) achieves better performance than (\ref{eq8}).
That is, the gap between (\ref{eq8}) and (\ref{eq17}) is tighter than that between (\ref{eq7}) and (\ref{eq17}) at low SNR.
At high SNR, the gap between (\ref{eq7}) and (\ref{eq17}) is tighter than that between (\ref{eq8}) and (\ref{eq17}).
Moreover, the difference between (\ref{eq7}) and (\ref{eq17}) becomes so small and it can be ignored.
Specifically, Table \ref{tab1} quantitatively shows the performance gaps between (\ref{eq7}) and (\ref{eq17}).
As can be seen, in the high SNR regime, the performance gap for each scenario is about 0.048 nats/transmission, which indicates that the asymptotic upper and lower bounds on secrecy capacity do not coincide. However, the difference is so small so that it can be ignored.
This conclusion coincides with that in \emph{Corollary \ref{cor1}}.

\begin{table}[!h]
\caption{Performance gaps between (\ref{eq7}) and (\ref{eq17}) at high SNR in Fig. \ref{fig3}.}
\begin{center}
\begin{tabularx}{8.5cm}{|p{1cm}|X|X|X|X|X|X|}\hline\hline
\centering \textbf{$P$ (dB)} &\multicolumn{3}{c|}{{\textbf{Performance gaps (nats/transmission)}}}
\tabularnewline\cline{2-4}
\centering  &\centering ${H_{\rm{B}}}=3000{H_{\rm{E}}}$ &\centering ${H_{\rm{B}}}=300{H_{\rm{E}}}$ &\centering ${H_{\rm{B}}}=30{H_{\rm{E}}}$
\tabularnewline\hline
\centering 65 &\centering 0.11118 &\centering 0.05804 &\centering 0.04971
\tabularnewline\hline
\centering 70 &\centering 0.07356 &\centering 0.05198 &\centering 0.04888
\tabularnewline\hline
\centering 75 &\centering 0.05803 &\centering 0.04971 &\centering 0.04858
\tabularnewline\hline
\centering 80 &\centering 0.05198 &\centering 0.04888 &\centering 0.04847
\tabularnewline\hline
\centering 85 &\centering 0.04871 &\centering 0.04858 &\centering 0.04844
\tabularnewline\hline\hline
\end{tabularx}
\end{center}
\label{tab1}
\end{table}

Figure \ref{fig4} shows the secrecy capacity bounds versus $\xi$ when $(a,b,c)=(5{\rm m}, 5{\rm m}, 3{\rm m})$, $(d,e,f)=(5{\rm m},4.5{\rm m},0{\rm m})$ and ${H_{\rm{B}}}/{H_{\rm{E}}}=300$.
As can be seen, all the secrecy capacity bounds are monotonically non-decreasing functions with respect to $\xi$.
When $P=35$ dB (i.e., at low SNR), the secrecy capacity bounds increase rapidly, while the increasing tendency becomes gradual when $P=65$ dB (i.e., at high SNR).
This indicates that, for VLC with only an average optical intensity constraint, the dimming target has a strong impact on secrecy capacity performance at low SNR. However, at high SNR, the effect of the dimming target becomes weak.
Moreover, at low SNR, the gap between (\ref{eq8}) and (\ref{eq17}) is smaller than that between (\ref{eq7}) and (\ref{eq17}).
At high SNR, the opposite is the case. This conclusion consists with that in Fig. \ref{fig3}.

\begin{figure}
\centering
\includegraphics[width=8.5cm]{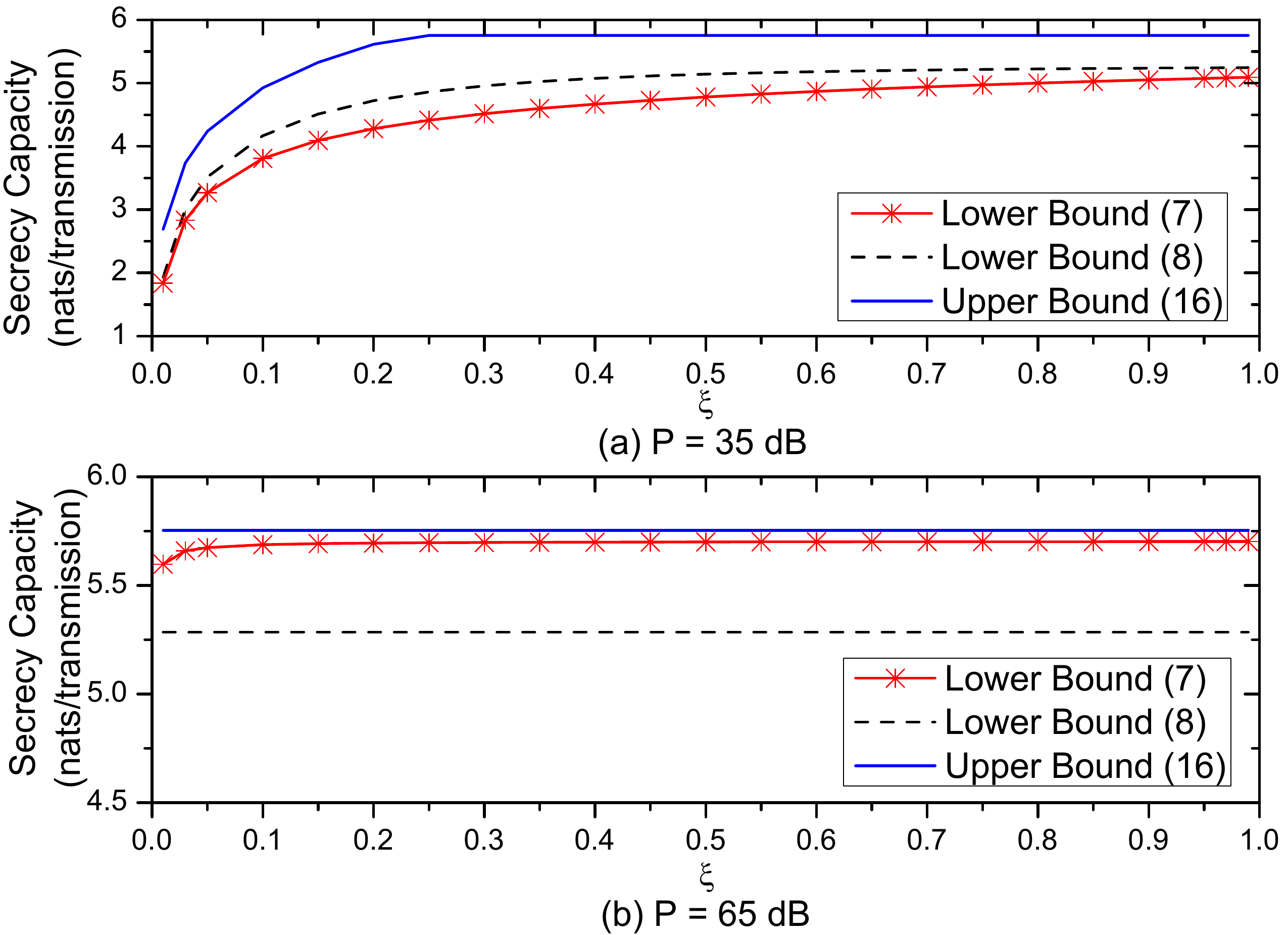}
\caption{Secrecy capacity bounds versus $\xi$ when $(a,b,c)=(5{\rm m}, 5{\rm m}, 3{\rm m})$, $(d,e,f)=(5{\rm m},4.5{\rm m},0{\rm m})$ and ${H_{\rm{B}}}/{H_{\rm{E}}}=300$.}
\label{fig4}
\end{figure}

Figure \ref{fig4_1} shows the secrecy capacity bounds versus ${H_{\rm{B}}}/{H_{\rm{E}}}$ when $\xi=0.2$.
It can be observed that the secrecy capacity bounds are zero when ${H_{\rm{B}}}/{H_{\rm{E}}}<1$, which indicates that
the information-theoretic security cannot be achieved.
When ${H_{\rm{B}}}/{H_{\rm{E}}}>1$, the secrecy capacity bounds are larger than zero.
The zone that ${H_{\rm{B}}}/{H_{\rm{E}}}>1$ is named as the available zone. In this zone, the secure transmission can be guaranteed.
Moreover, the secrecy capacity bounds increase with the increase of ${H_{\rm{B}}}/{H_{\rm{E}}}$.
Furthermore, eq. (\ref{eq8}) is tighter than (\ref{eq7}) at low SNR, while the performance of (\ref{eq7}) is better than that of (\ref{eq8}) at high SNR. Similar conclusion can be drawn from Fig. \ref{fig4}.

\begin{figure}
\centering
\includegraphics[width=9cm]{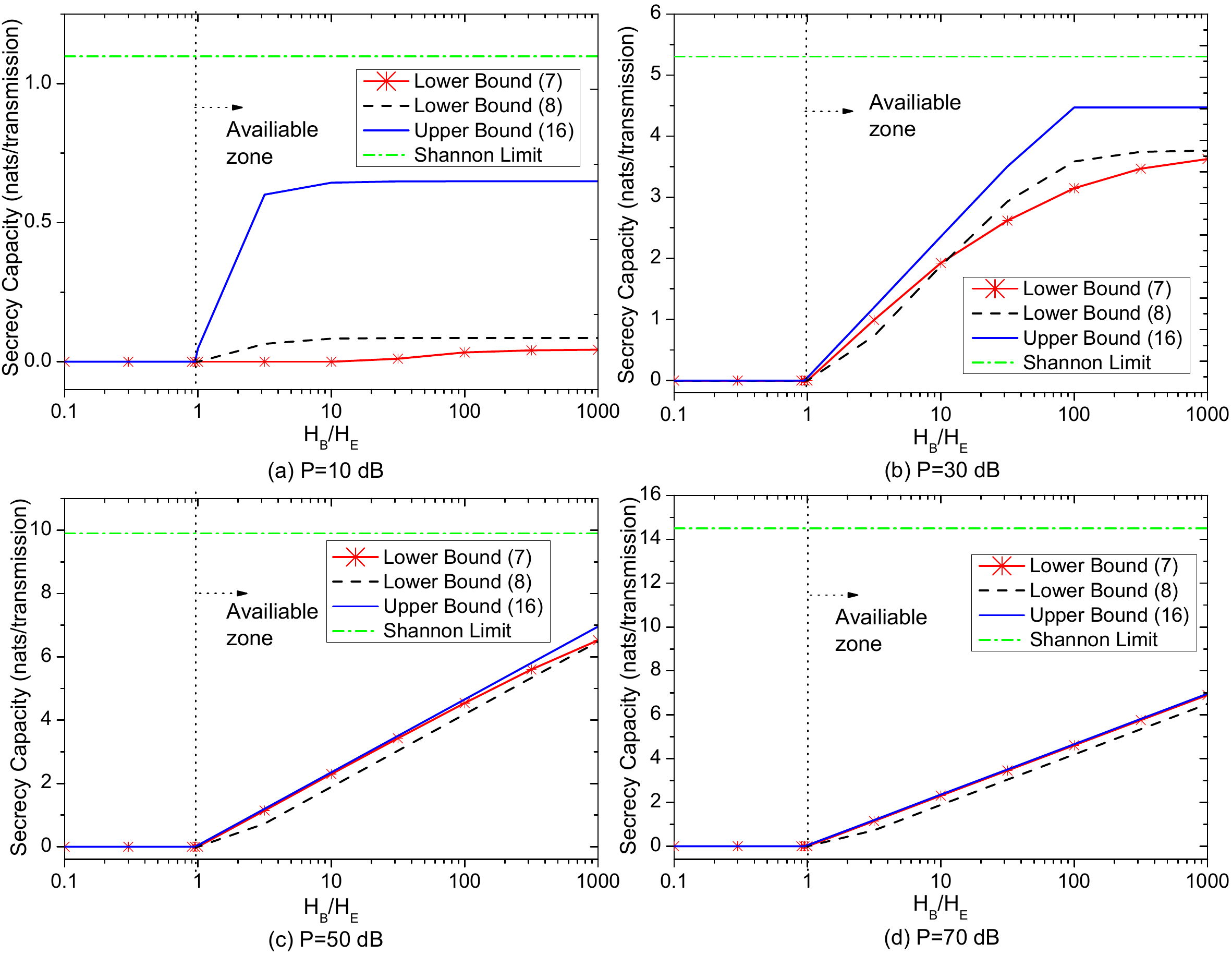}
\caption{Secrecy capacity bounds versus ${H_{\rm{B}}}/{H_{\rm{E}}}$ when $\xi=0.2$.}
\label{fig4_1}
\end{figure}

\subsection{Results of VLC with Both Average and Peak Optical Intensity Constraints}
To verify the accuracy of the lower bounds (\ref{eq19}), (\ref{eq20}) and the upper bound (\ref{eq26}), Figs. \ref{fig5}-\ref{fig6_1} and Table \ref{tab2} are provided in this subsection.

Figure \ref{fig5} shows the secrecy capacity bounds versus $A$ with different ${H_{\rm{B}}}/{H_{\rm{E}}}$ when $\xi=0.2$, $P=A$, $(a,b,c)=(5{\rm m}, 5{\rm m}, 3{\rm m})$ and $(d,e,f)=(5{\rm m},4.5{\rm m},0{\rm m})$. In this figure, the Shannon limit is also provided. Once again, the performance of Shannon limit always outperforms that of all secrecy capacity bounds.
Similar to Fig. \ref{fig3}, the secrecy capacity bounds increase and then tend to stable values with the increase of $A$.
When $A$ is small (i.e., at low SNR), the value of (\ref{eq20}) is larger than that of (\ref{eq19}).
However, when $A$ is large (i.e., at high SNR), the performance of (\ref{eq19}) is better than that of (\ref{eq20}).
At high SNR, the difference between (\ref{eq19}) and (\ref{eq26}) is very small and can be ignored.
Different from Table \ref{tab1}, the performance gaps between (\ref{eq19}) and (\ref{eq26}) in Table \ref{tab2} are almost zero, which coincides with the conclusion in \emph{Corollary \ref{cor2}}.
This indicates that the asymptotic upper and lower bounds on secrecy capacity coincide at high SNR.

\begin{figure}
\centering
\includegraphics[width=8.5cm]{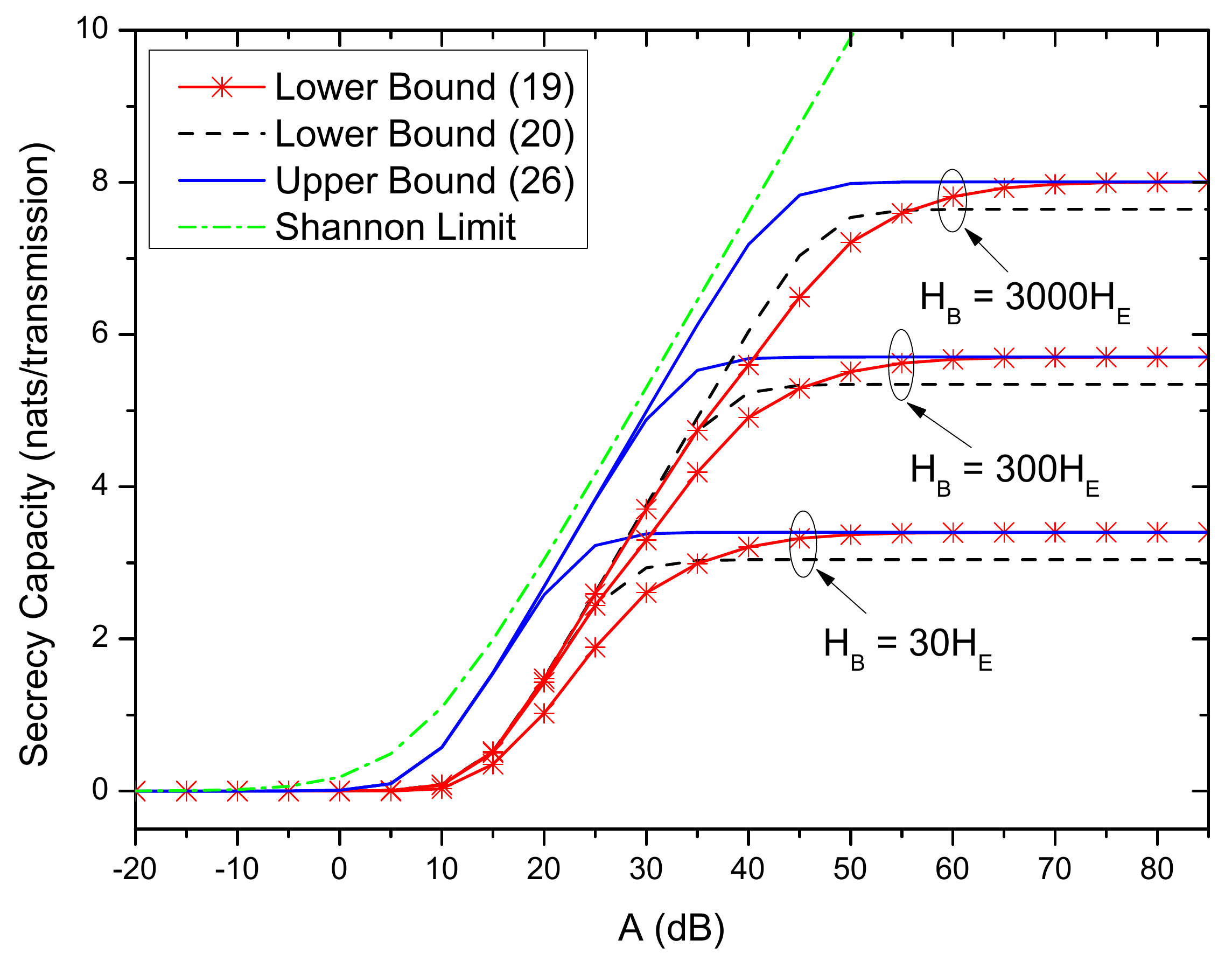}
\caption{Secrecy capacity bounds versus $A$ with different ${H_{\rm{B}}}/{H_{\rm{E}}}$ when $\xi=0.2$, $P=A$, $(a,b,c)=(5{\rm m}, 5{\rm m}, 3{\rm m})$ and $(d,e,f)=(5{\rm m},4.5{\rm m},0{\rm m})$.}
\label{fig5}
\end{figure}

\begin{table}[!h]
\caption{Performance gaps between (\ref{eq19}) and (\ref{eq26}) at high SNR.}
\begin{center}
\begin{tabularx}{8.5cm}{|p{1cm}|X|X|X|X|X|X|}\hline\hline
\centering \textbf{$A$ (dB)} &\multicolumn{3}{c|}{{\textbf{Performance gaps (nats/transmission)}}}
\tabularnewline\cline{2-4}
\centering  &\centering ${H_{\rm{B}}}=3000{H_{\rm{E}}}$ &\centering ${H_{\rm{B}}}=300{H_{\rm{E}}}$ &\centering ${H_{\rm{B}}}=30{H_{\rm{E}}}$
\tabularnewline\hline
\centering 65 &\centering 8.0331$\times 10^{-2}$ &\centering 1.1486$\times 10^{-2}$ &\centering  1.4744$\times 10^{-3}$
\tabularnewline\hline
\centering 70 &\centering 3.1006$\times 10^{-2}$ &\centering 4.1499$\times 10^{-3}$ &\centering  5.1763 $\times 10^{-4}$
\tabularnewline\hline
\centering 75 &\centering 1.1479$\times 10^{-2}$ &\centering 1.4750$\times 10^{-3}$ &\centering 1.7993$\times 10^{-4}$
\tabularnewline\hline
\centering 80 &\centering 4.1472$\times 10^{-3}$ &\centering 5.1785$\times 10^{-4}$ &\centering 6.2036$\times 10^{-5}$
\tabularnewline\hline
\centering 85 &\centering 1.4741$\times 10^{-3}$ &\centering 1.8001$\times 10^{-4}$ &\centering 2.1242$\times 10^{-5}$
\tabularnewline\hline\hline
\end{tabularx}
\end{center}
\label{tab2}
\end{table}

Figure \ref{fig6} shows the secrecy capacity bounds versus $\xi$ when ${H_{\rm{B}}}/{H_{\rm{E}}}=300$, $P=A$, $(a,b,c)=(2.5{\rm m}, 2.5{\rm m}, 3{\rm m})$ and $(d,e,f)=(2.5{\rm m},2{\rm m},0.8{\rm m})$.
Obviously, the dimming target has a strong impact on system performance.
For (\ref{eq19}) and (\ref{eq26}), when $A=35$ dB, the secrecy capacity bounds increase rapidly with $\xi$; when $A=65$ dB, the impact of $\xi$ is slight, and the secrecy capacity bounds almost achieve stable values.
However, for (\ref{eq20}), the curve is completely symmetric with respect to $\xi=0.5$,
and the maximum value is obtained when $\xi=0.5$.
It can be found from Fig. \ref{fig6}(a) that (\ref{eq20}) is better than (\ref{eq19}) at low SNR.
However, at high SNR as shown in Fig. \ref{fig6}(b), it is better to employ (\ref{eq19}) as the lower bound.
Similar results can also be seen in Fig. \ref{fig5}.

\begin{figure}
\centering
\includegraphics[width=8.5cm]{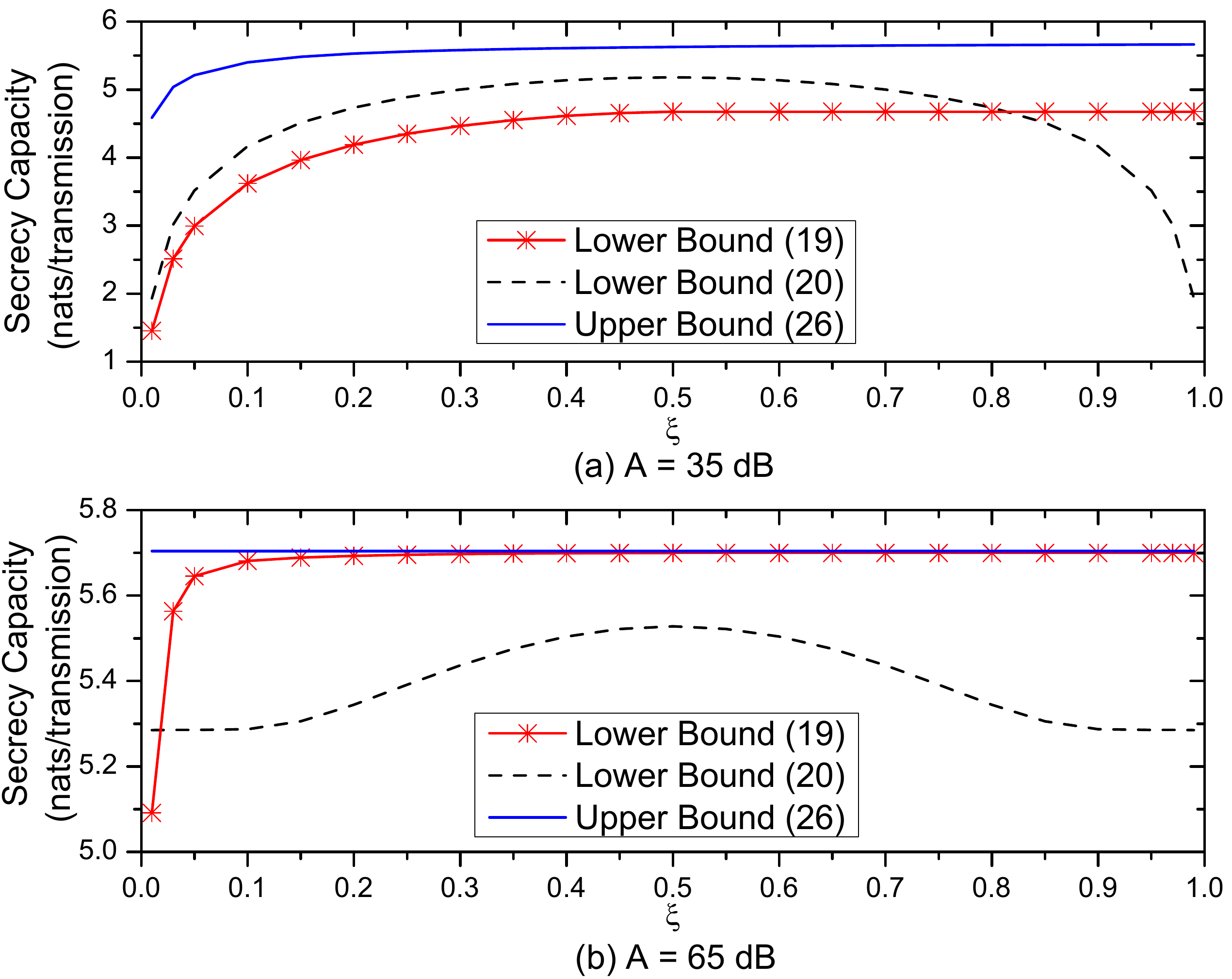}
\caption{Secrecy capacity bounds versus $\xi$ when ${H_{\rm{B}}}/{H_{\rm{E}}}=300$, $P=A$, $(a,b,c)=(2.5{\rm m}, 2.5{\rm m}, 3{\rm m})$ and $(d,e,f)=(2.5{\rm m},2{\rm m},0.8{\rm m})$.}
\label{fig6}
\end{figure}

Figure \ref{fig6_1} shows secrecy capacity bounds versus ${H_{\rm{B}}}/{H_{\rm{E}}}$ when $\xi=0.2$ and $P=A$.
Here, the Shannon limit is also provided. As can be seen, the performance of Shannon limit always outperforms that of all secrecy capacity bounds.
Obviously, with the increase of ${H_{\rm{B}}}/{H_{\rm{E}}}$, the secrecy performance improves.
At the available zone, the secure communications can be implemented.
Moreover, the values of lower bound (\ref{eq20}) are almost larger than that of (\ref{eq19}) at low SNR, and the performance of (\ref{eq19}) is better than that of (\ref{eq20}) at high SNR. This observation is the same as that in Fig. \ref{fig6}.

\begin{figure}
\centering
\includegraphics[width=9cm]{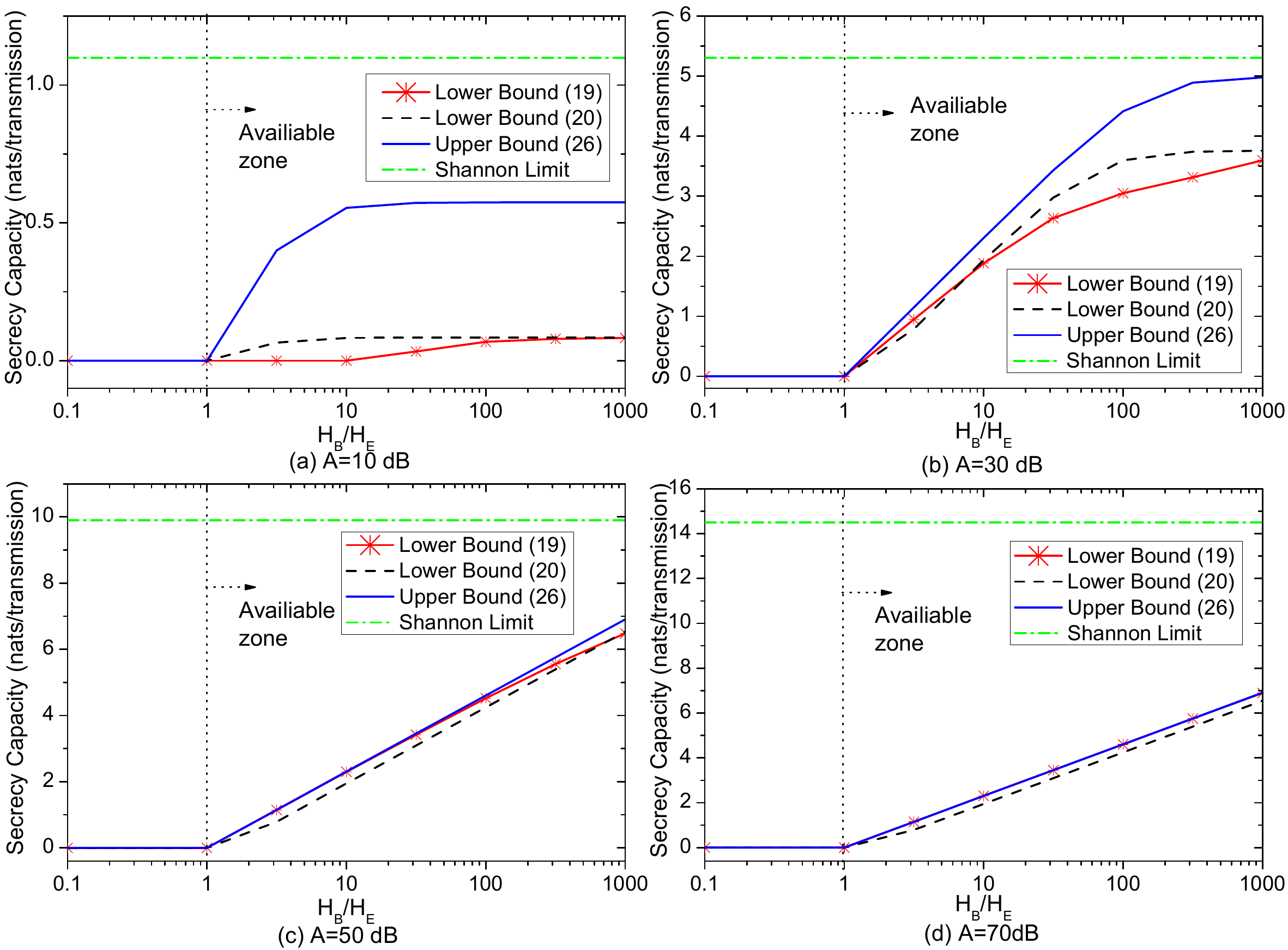}
\caption{Secrecy capacity bounds versus ${H_{\rm{B}}}/{H_{\rm{E}}}$ when $\xi=0.2$ and $P=A$.}
\label{fig6_1}
\end{figure}

\subsection{Results of Insecure Region}
As it is known, when the main channel is worse than the eavesdropping channel (i.e., $H_{\rm B}/\sigma_{\rm B} < H_{\rm E}/\sigma_{\rm E}$), the secrecy capacity is zero, and thus secure transmission cannot be guaranteed.
Here, the ``insecure region" \cite{BIBadd} is analyzed,
which represents the receive region that the secrecy capacity is zero.
If $\sigma_{\rm B}=\sigma_{\rm E}$, the indoor system is insecure when $H_{\rm B} < H_{\rm E}$.
According to (\ref{eq5}), it can be known that the insecure region is a disc with center $(a, b, f)$ and radius $\sqrt {{{(a - d)}^2} + {{(b - e)}^2}} $.

\begin{figure*}[!t]
  \centering
  \subfigure[$(d, e, f)=(5{\rm m}, 5{\rm m}, 0{\rm m})$]{
    \label{fig:subfig:a}
    \includegraphics[width=0.25\textwidth]{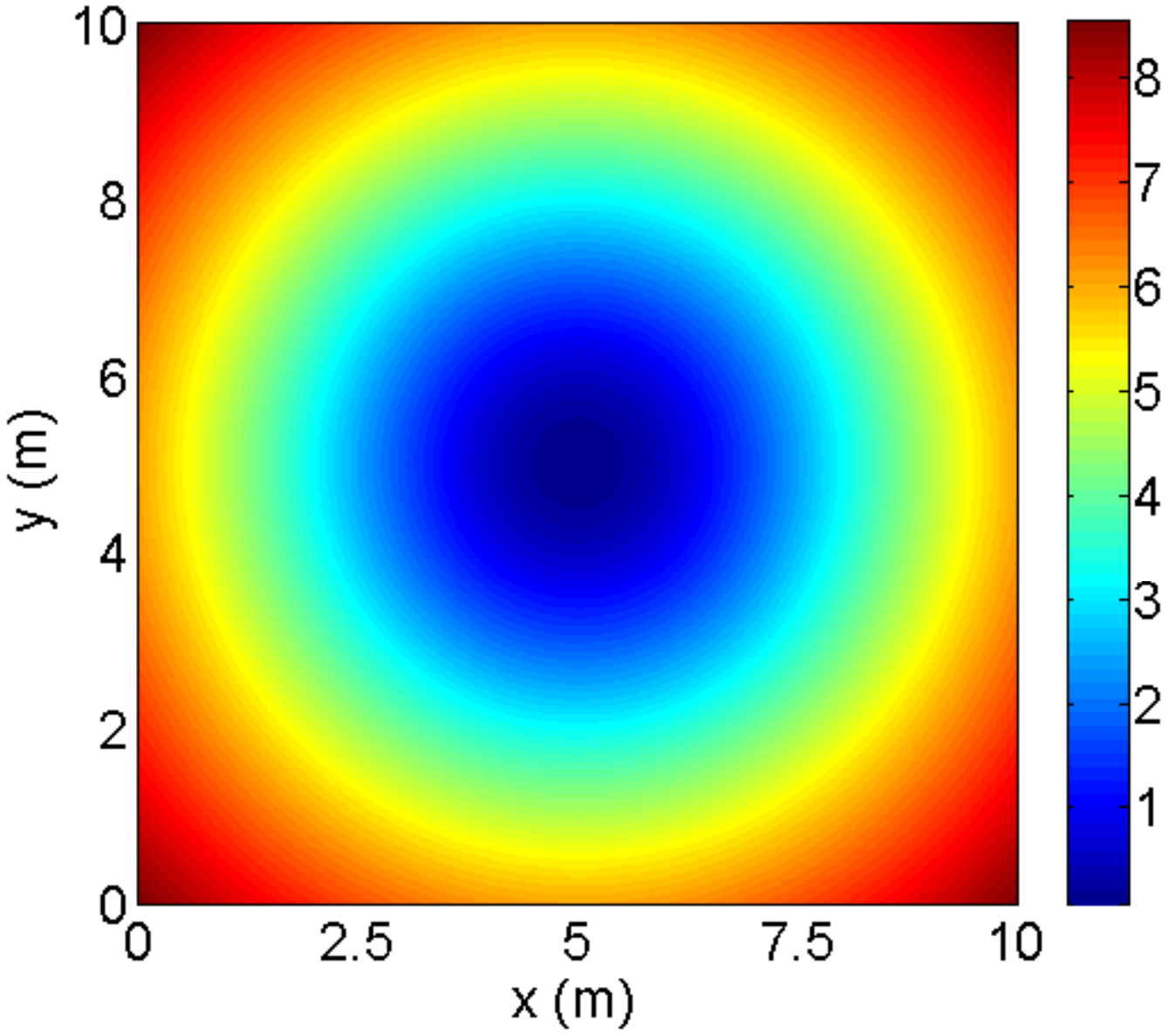}}
  \subfigure[$(d, e, f)=(3{\rm m}, 3{\rm m}, 0{\rm m})$]{
    \label{fig:subfig:b}
    \includegraphics[width=0.25\textwidth]{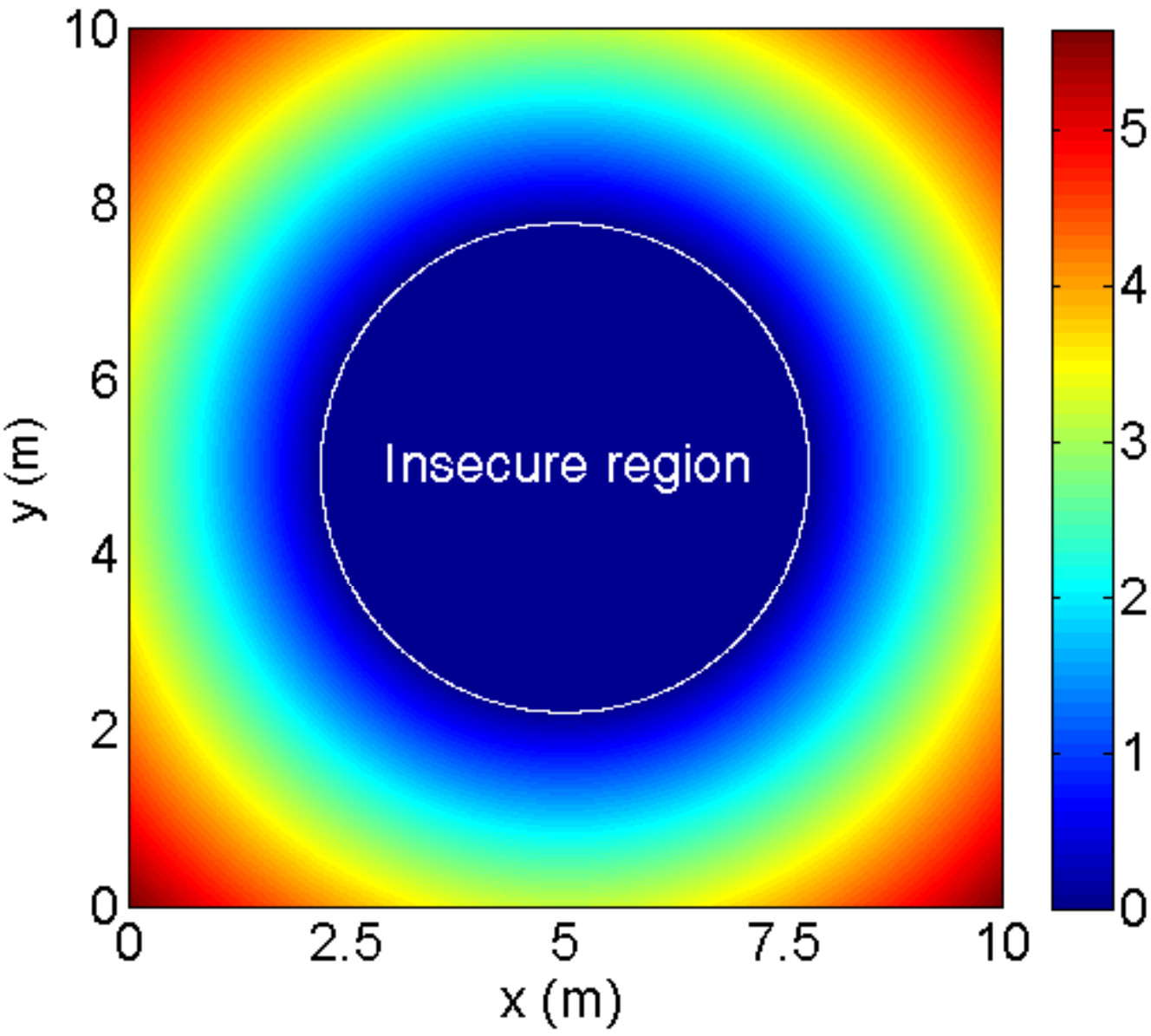}}
  \subfigure[$(d, e, f)=(0{\rm m}, 0{\rm m}, 0{\rm m})$]{
    \label{fig:subfig:b}
    \includegraphics[width=0.27\textwidth]{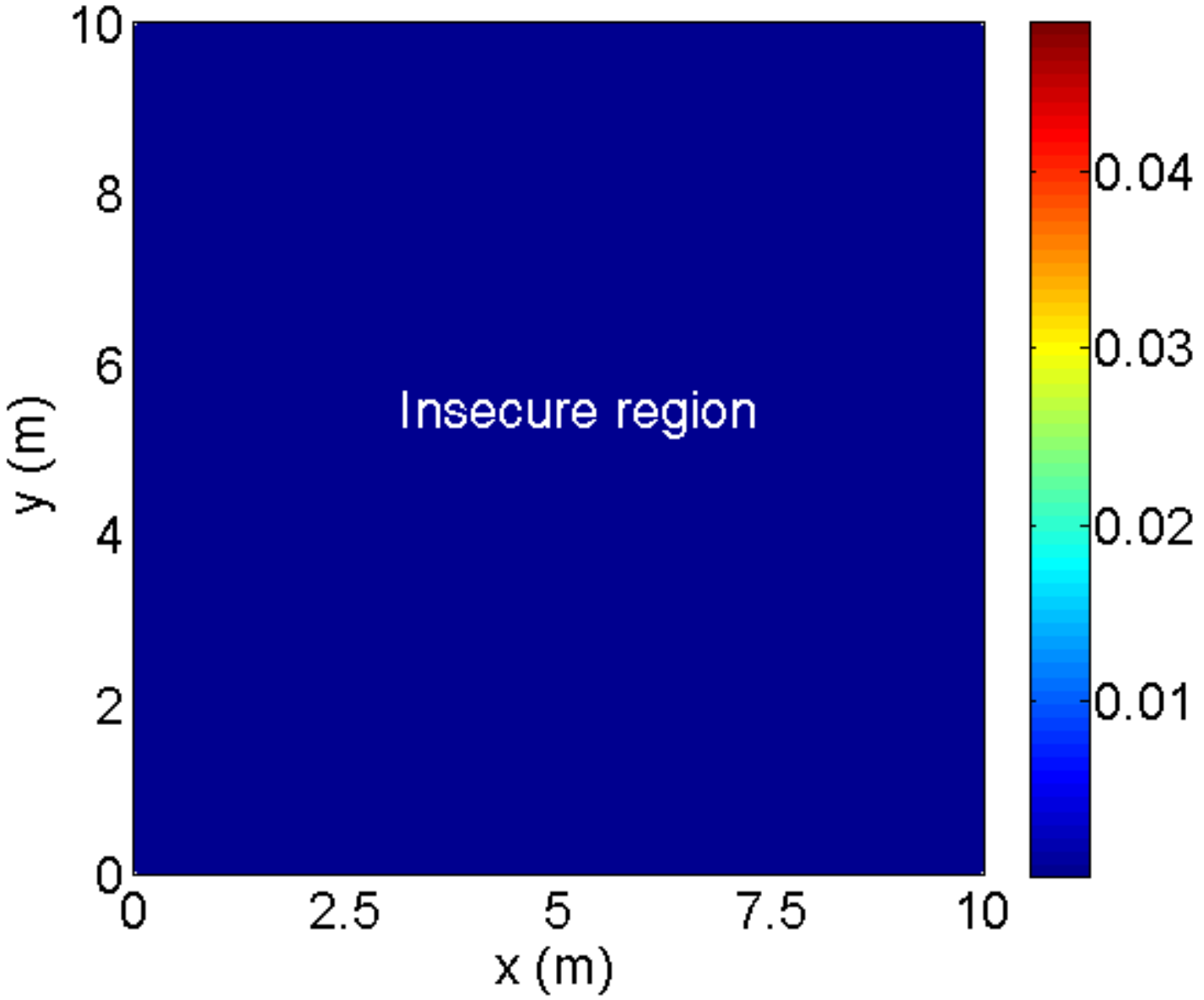}}
  \caption{Secrecy capacity bound (\ref{eq17}) versus different positions of Eve when $\xi {\rm{ = 0}}{\rm{.2}}$, $P=50$ dB and $(a, b, c)=(5{\rm m}, 5{\rm m}, 3{\rm m})$.}
  \label{figA}
\end{figure*}

\begin{figure*}[!t]
  \centering
  \subfigure[$(d, e, f)=(5{\rm m}, 5{\rm m}, 0{\rm m})$]{
    \label{fig:subfig:a}
    \includegraphics[width=0.25\textwidth]{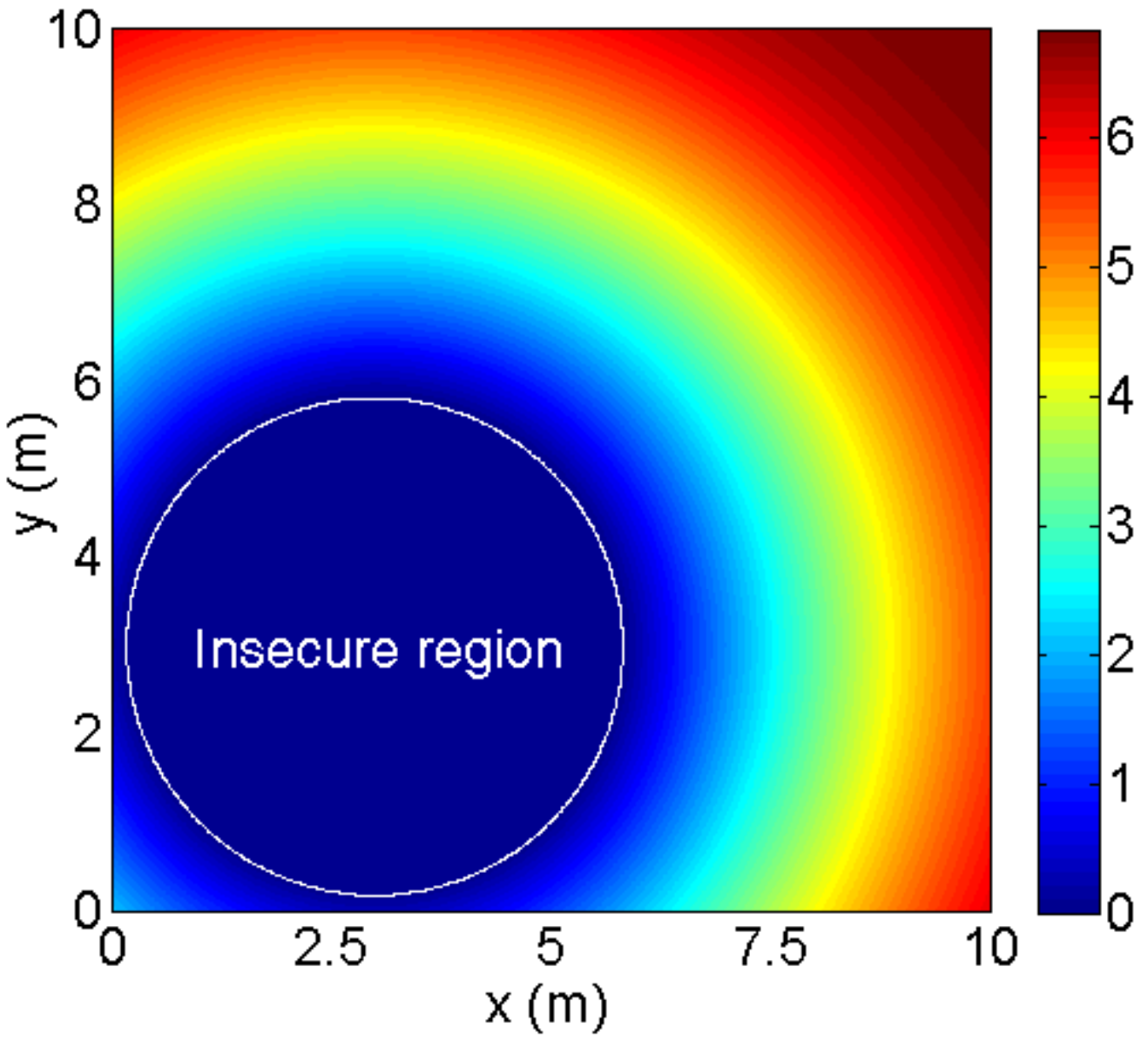}}
  \subfigure[$(d, e, f)=(3{\rm m}, 3{\rm m}, 0{\rm m})$]{
    \label{fig:subfig:b}
    \includegraphics[width=0.25\textwidth]{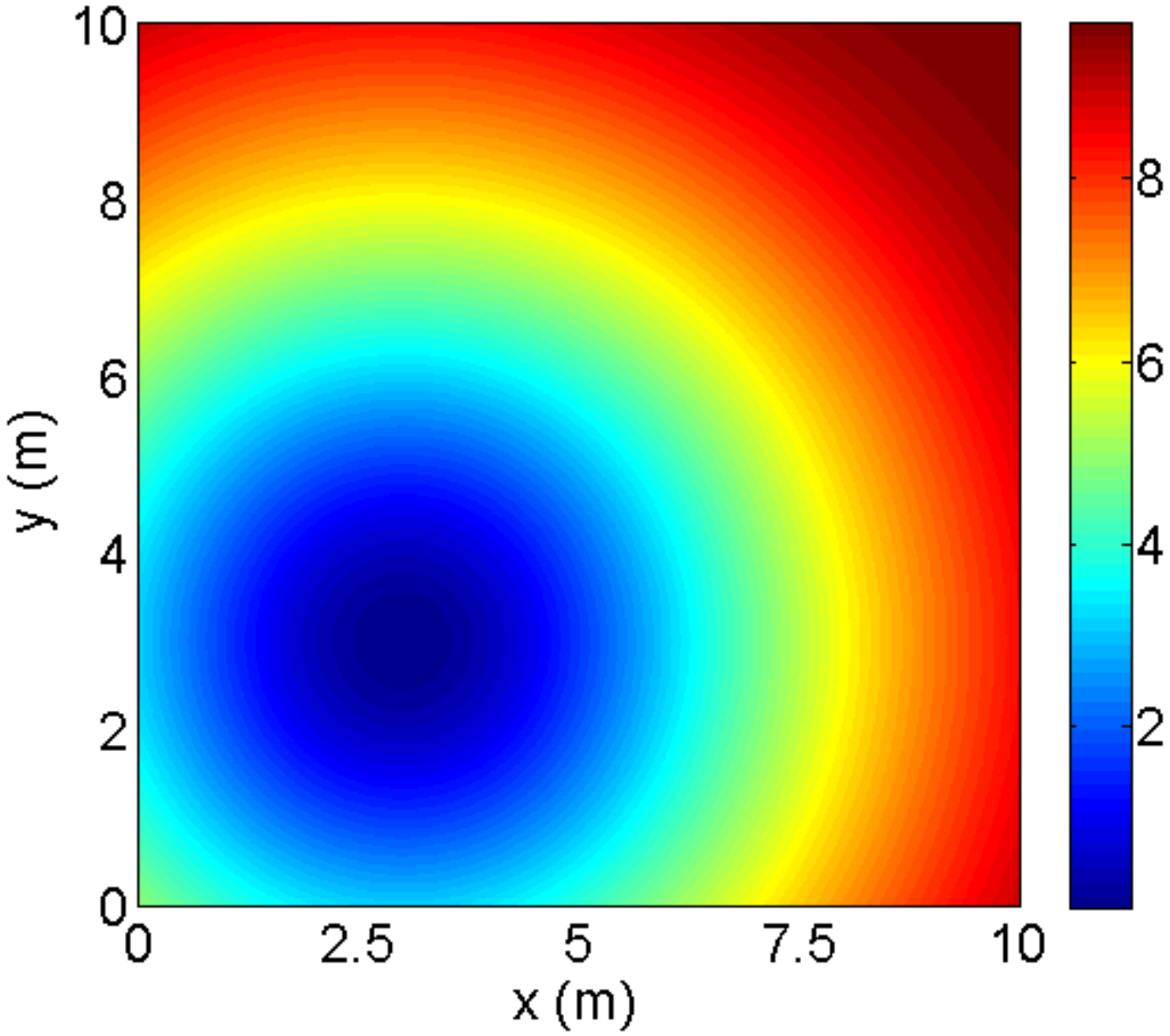}}
   \subfigure[$(d, e, f)=(0{\rm m}, 0{\rm m}, 0{\rm m})$]{
    \label{fig:subfig:b}
    \includegraphics[width=0.25\textwidth]{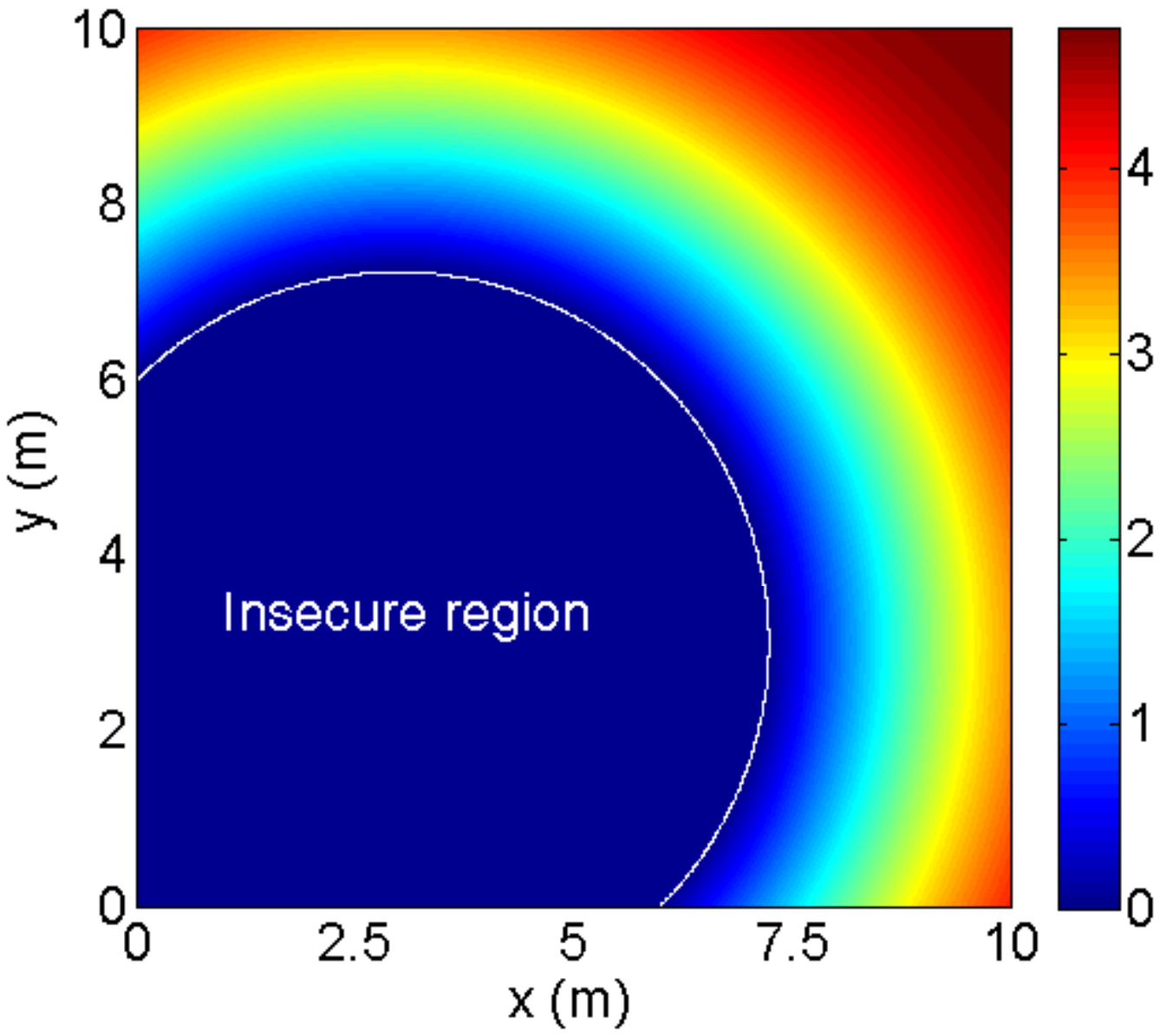}}
  \caption{Secrecy capacity bound (\ref{eq26}) versus different positions of Eve when $\xi {\rm{ = 0}}{\rm{.2}}$, $A=P=50$ dB and $(a, b, c)=(3{\rm m}, 3{\rm m}, 3{\rm m})$.}
  \label{figB}
\end{figure*}

To verify the insecure region,
Fig. \ref{figA} shows secrecy capacity bound (\ref{eq17}) for different positions of Eve when $\xi= 0.2$, $P=50$ dB and $(a, b, c)=(5{\rm m}, 5{\rm m}, 3{\rm m})$.
In Fig. \ref{figA}, the dark blue area represents the insecure region.
As shown in Fig. \ref{figA}(a),
when Bob is located underneath Alice, i.e., $(d, e, f)=(5{\rm m}, 5{\rm m}, 0{\rm m})$,
no area in the receiver plane belongs to the insecure region.
That is, the secure transmission can be guaranteed for all positions of the receiver plane.
When Bob moves to the coordinate $(3{\rm m}, 3{\rm m}, 0{\rm m})$,
the insecure region enlarges.
Moreover, when Bob moves to the corner, i.e., $(d, e, f)=(0{\rm m}, 0{\rm m}, 0{\rm m})$,
all areas become the insecure region.
By observing Fig. \ref{figA}, the insecure region is a disc with center $(5{\rm m}, 5{\rm m}, 0{\rm m})$ and radius $\sqrt {{{(5- d)}^2} + {{(5 - e)}^2}}$m.

Fig. \ref{figB} shows secrecy capacity bound (\ref{eq26}) for different positions of Eve when $\xi {\rm{ = 0}}{\rm{.2}}$, $A=P=50$ dB and $(a, b, c)=(3{\rm m}, 3{\rm m}, 3{\rm m})$.
In Fig. \ref{figB}(a), when Bob is located at $(5{\rm m}, 5{\rm m}, 0{\rm m})$,
the insecure region is large.
When Bob is located underneath Alice in Fig. \ref{figB}(b),
the insecure region seems to vanish.
When Bob moves to the corner in Fig. \ref{figB}(c),
the insecure region enlarges once again.
From Fig. \ref{figB}, the insecure region is a disc with center $(3{\rm m}, 3{\rm m}, 0{\rm m})$ and radius $\sqrt {{{(3- d)}^2} + {{(3- e)}^2}}$m.

\section{Conclusions}
\label{section6}
Unlike conventional RFWC, the indoor VLC is well modelled with optical intensity constraints imposed on the channel input.
Therefore, the PHY security in VLC is different from that in RFWC.
In this paper, we have investigated the secrecy capacity for indoor VLC.
Two scenarios are considered, i.e., one is only with an average optical intensity constraint,
and the other is with both average and peak optical intensity constraints.
Closed-form expressions of the lower and upper bounds on secrecy capacity are derived.
It is shown that the gap between the lower and upper bounds is small, which verifies the accuracy of the derived expressions.
Moreover, at high SNR, when only considering the average optical intensity constraint,
the asymptotic lower and upper bounds do not coincide but with a small performance gap (i.e., 0.048 nats/transmission).
When both average and peak optical intensity constraints are considered, the asymptotic lower and upper bounds coincide,
and thus the secrecy capacity can be obtained precisely.

After obtaining the secrecy capacity bounds for the indoor VLC,
exploring the schemes to enhance the PHY security is the natural next step.
As future research directions, it is of interest to seek new power, code, channel,
and signal detection approaches to improve the PHY security for indoor VLC.

\numberwithin{equation}{section}
\appendices
\section{Proof of lower bound (\ref{eq7}) in Theorem \ref{them1}}
\label{appa}
\renewcommand{\theequation}{A.\arabic{equation}}
For two arbitrary functions ${f_1}(x)$ and ${f_2}(x)$, we have
$\mathop {\max }\limits_x \left( {{f_1}(x) - {f_2}(x)} \right) \ge \mathop {\max }\limits_x {f_1}(x) - \mathop {\max }\limits_x {f_2}(x)$ \cite{BIB21}.
Therefore, the secrecy capacity in (\ref{eq6}) is lower-bounded by
\begin{eqnarray}
C_{\rm s} \ge \mathop {\max }\limits_{{f_X}(x)} I\left( {X;{Y_{\rm{B}}}} \right) - \mathop {\max }\limits_{{f_X}(x)} I\left( {X;{Y_{\rm{E}}}} \right)
 \buildrel \Delta \over = {C_{\rm{B}}} - {C_{\rm{E}}}.
  \label{eq28}
\end{eqnarray}
According to \cite{BIB22}, a lower bound of ${C_{\rm{B}}}$ can be easily obtained as (12) in \cite{BIB20_add3}, and an upper bound of ${C_{\rm{E}}}$ can be easily obtained as (13) in \cite{BIB20_add3}.
Then, submit (12) and (13) in \cite{BIB20_add3} into (\ref{eq28}), eq. (\ref{eq7}) can be derived.

\section{Proof of lower bound (\ref{eq8}) in Theorem \ref{them1}}
\label{appb}
\renewcommand{\theequation}{B.\arabic{equation}}
According to (29) in \cite{BIB20} and the EPI (9.181) in \cite{BIB23}, the objective function in (\ref{eq6}) can be lower-bounded by
\begin{eqnarray}
C_{\rm s}\!\!\!\!\! &\ge&\!\!\!\!\! \mathop {\max }\limits_{{f_X}\!(x)} \!\!\left\{\! {\frac{1}{2}\!\ln \!\left[\! {{e^{2\left[ {{\cal H}(X) \!+\! \ln \left( {{H_{\rm{B}}}} \right)} \right]}} \!+\! 2\pi e\sigma _{\rm{B}}^2} \!\right] \!-\! \frac{1}{2}\!\ln \!\left[ {2\pi e{\mathop{\rm var}} ( {Y_{\rm E}}) } \right]}\! \right\} \nonumber\\
&+&\!\!\!\!\! \ln \left( {\frac{{{\sigma _{\rm{E}}}}}{{{\sigma _{\rm{B}}}}}} \right).
\label{eq37}
\end{eqnarray}
Obviously, a lower bound on the secrecy capacity can be derived by dropping the maximization and choosing an arbitrary ${f_X}(x)$ under the given constraints in (\ref{eq6}). Without loss of generality, we choose an input PDF that maximizes the entropy ${\cal H}(X)$ under the constraints in (\ref{eq6}). Such an input PDF can be found by solving the following functional optimization problem
\begin{eqnarray}
&& \mathop {\min }\limits_{{f_X}(x)} {\cal J}\left[ {{f_X}(x)} \right]\triangleq \int_0^\infty  {{f_X}(x) \ln} \left[ {{f_X}(x)} \right]{\rm{d}}x \nonumber \\
{\rm{s.t.}} && \int_0^\infty  {{f_X}(x){\rm d}x}  = 1 \nonumber \\
&& \int_0^\infty  {x{f_X}(x){\rm d} x}  = \xi P.
\label{eq39}
\end{eqnarray}
Note that problem (\ref{eq39}) can be solved by using the variational method.
Referring to our previous paper \cite{BIB20_add3}, the PDF of $f_X(x)$ is obtained as
\begin{equation}
{f_X}(x) = \frac{1}{{\xi P}}{e^{ - \frac{1}{{\xi P}}x}},\;x \ge 0.
\label{eq48}
\end{equation}
Furthermore, ${\cal H}(X)$ and ${\mathop{\rm var}} ({Y_{\rm{E}}})$ can be written as
\begin{eqnarray}
\left\{ \begin{array}{l}
{\cal H}(X) = \ln \left( {e\xi P} \right) \\
{\mathop{\rm var}} ({Y_{\rm{E}}}) = H_{\rm{E}}^2{\xi ^2}{P^2} + \sigma _{\rm{E}}^2.
\end{array} \right.
\label{eq50}
\end{eqnarray}
Therefore, submitting (\ref{eq50}) into (\ref{eq37}), eq. (\ref{eq8}) can be derived.

\section{Proof of upper bound (\ref{eq17}) in Theorem \ref{them2}}
\label{appc}
\renewcommand{\theequation}{C.\arabic{equation}}
According to (\ref{eq16}), we have (\ref{eq51}) as shown at the top of the next page.
\begin{table*}\normalsize
\begin{eqnarray}
C_{\rm s} &\le& \underbrace {{E_{{X^*}}}\left\{ {\int_{ - \infty }^\infty  {\int_{ - \infty }^\infty  {{f_{{Y_{\rm{B}}}{Y_{\rm E}}|X}}({y_{\rm{B}}},{y_{\rm{E}}}|X)\ln \left[ {{f_{{Y_{\rm{B}}}|X{Y_{\rm E}}}}({y_{\rm{B}}}|X,{y_{\rm{E}}})} \right]{\rm{d}}{y_{\rm{B}}}} {\rm{d}}{y_{\rm{E}}}} } \right\}}_{{I_1}} \nonumber\\
&&\underbrace { - {E_{{X^*}}}\left\{ {\int_{ - \infty }^\infty  {\int_{ - \infty }^\infty  {{f_{{Y_{\rm{B}}}{Y_{\rm E}}|X}}({y_{\rm{B}}},{y_{\rm{E}}}|X)\ln \left[ {{g_{{Y_{\rm{B}}}|{Y_{\rm E}}}}({y_{\rm{B}}}|{y_{\rm{E}}})} \right]{\rm{d}}{y_{\rm{B}}}} {\rm{d}}{y_{\rm{E}}}} } \right\}}_{{I_2}}.
\label{eq51}
\end{eqnarray}
\hrulefill
\end{table*}
Moreover, $I_1$ in (\ref{eq51}) can be written as
\begin{eqnarray}
{I_1} =  - \left[ {{\cal H}({Y_{\rm{B}}}|{X^*}) + {\cal H}({Y_{\rm{E}}}|{X^*},{Y_{\rm{B}}}) - {\cal H}({Y_{\rm{E}}}|{X^*})} \right],
 \label{eq52}
\end{eqnarray}
where ${\cal H}({Y_{\rm{B}}}|X^*)$ is given by
\begin{eqnarray}
{\cal H}\left( {{Y_{\rm{B}}}\left| {{X^*}} \right.} \right) ={\cal H}\left( {{Y_{\rm{B}}}\left| {{X}} \right.} \right)= \frac{1}{2}{\rm{ln}}\left( {2\pi e\sigma _{\rm{B}}^2} \right).
 \label{eq53}
\end{eqnarray}
Similarly, ${\cal H}({Y_{\rm{E}}}|{X^*})$ can be derived as
\begin{equation}
{\cal H}({Y_{\rm{E}}}|{X^*}) = \frac{1}{2}{\rm{ln}}\left( {2\pi e\sigma _{\rm{E}}^2} \right).
\label{eq54}
\end{equation}
Furthermore, ${\cal H}({Y_{\rm{E}}}|{X^*},{Y_{\rm{B}}})$ can be expressed as
\begin{eqnarray}
{\cal H}({Y_{\rm{E}}}|{X^*},{Y_{\rm{B}}})=\frac{1}{2}{\rm{ln}}\left[ {2\pi e\left( {\frac{{H_{\rm{E}}^2}}{{H_{\rm{B}}^2}}\sigma _{\rm{B}}^2 + \sigma _{\rm{E}}^2} \right)} \right].
 \label{eq57}
\end{eqnarray}

Substituting (\ref{eq53}), (\ref{eq54}) and (\ref{eq57}) into (\ref{eq52}), $I_1$ can be finally written as
\begin{eqnarray}
{I_1}  =  - \frac{1}{2}{\rm{ln}}\left[ {2\pi e\sigma _{\rm{B}}^2\left( {1 + \frac{{H_{\rm{E}}^2\sigma _{\rm{B}}^2}}{{H_{\rm{B}}^2\sigma _{\rm{E}}^2}}} \right)} \right].
\label{eq58}
\end{eqnarray}
The main challenge of deriving $I_2$ in (\ref{eq51}) is that the input can be arbitrarily large without a peak optical intensity constraint.
That makes it much harder to find a bound on expression like ${E_{{X^*}}}({X^2})$. To obtain $I_2$, ${g_{{Y_{\rm{B}}}|{Y_{\rm E}}}}({y_{\rm{B}}}|{y_{\rm{E}}})$ is chosen as
\begin{equation}
{g_{{Y_{\rm{B}}}|{Y_{\rm E}}}}({y_{\rm{B}}}|{y_{\rm{E}}}) = \frac{1}{{2{s^2}}}{e^{ - \frac{{\left| {{y_{\rm{B}}} - \mu {y_{\rm{E}}}} \right|}}{{{s^2}}}}},
\label{eq59}
\end{equation}
where $\mu$ and $s$ are two free parameters to be determined.

Moreover, ${f_{{Y_{\rm{B}}}{Y_{\rm E}}|X}}({y_{\rm{B}}},{y_{\rm{E}}}|X)$ can be written as
\begin{eqnarray}
{f_{{Y_{\rm{B}}}{Y_{\rm E}}|X}}(\left. {{y_{\rm{B}}},{y_{\rm{E}}}} \right|X) \!=\! \frac{{{e^{ - \frac{{{{({y_{\rm{B}}} \!-\! {H_{\rm{B}}}X)}^2}}}{{2\sigma _{\rm{B}}^2}}}}}}{{\sqrt {2\pi } {\sigma _{\rm{B}}}}}\frac{{{e^{ - \frac{{{{\left( {{y_{\rm{E}}} \!-\! \frac{{{H_{\rm{E}}}}}{{{H_{\rm{B}}}}}{y_{\rm{B}}}} \right)}^2}}}{{2\left( {\frac{{H_{\rm{E}}^2}}{{H_{\rm{B}}^2}}\sigma _{\rm{B}}^2 \!+\! \sigma _{\rm{E}}^2} \right)}}}}}}{{\sqrt {2\pi \left( {\frac{{H_{\rm{E}}^2}}{{H_{\rm{B}}^2}}\sigma _{\rm{B}}^2 \!+\! \sigma _{\rm{E}}^2} \right)} }}.
 \label{eq60}
\end{eqnarray}
Therefore, $I_2$ can be written as (\ref{eq61}) as shown at the top of the next page.
\begin{table*}\normalsize
\begin{eqnarray}
{I_2} = \ln (2{s^2}) \!+\! \frac{1}{{{s^2}}}{E_{{X^*}}}\!\!\left[\! {\int_{ - \infty }^\infty \! {\frac{{{e^{ - \frac{{{{({y_{\rm{B}}} \!-\! {H_{\rm{B}}}X)}^2}}}{{2\sigma _{\rm{B}}^2}}}}}}{{\sqrt {2\pi } {\sigma _{\rm{B}}}}}\!\!\! \int_{ - \infty }^\infty \!\!\! {\frac{{{e^{ - \frac{{{t^2}}}{{2\left( {\frac{{H_{\rm{E}}^2}}{{H_{\rm{B}}^2}}\sigma _{\rm{B}}^2 + \sigma _{\rm{E}}^2} \right)}}}}}}{{\sqrt {2\pi \!\!\left(\! {\frac{{H_{\rm{E}}^2}}{{H_{\rm{B}}^2}}\sigma _{\rm{B}}^2 \!+\! \sigma _{\rm{E}}^2} \!\right)} }}\left|\! {\left(\! {1 \!\!-\!\! \mu \frac{{{H_{\rm{E}}}}}{{{H_{\rm{B}}}}}}\! \right){y_{\rm{B}}} \!\!-\!\! \mu t} \right|\!{\rm{d}}t{\rm{d}}{y_{\rm{B}}}} } }\!\! \right].
 \label{eq61}
\end{eqnarray}
\hrulefill
\end{table*}
Because $\left| {a - b} \right| \le \left| a \right| + \left| b \right|$ always holds, eq. (\ref{eq61}) can be upper-bounded by
\begin{eqnarray}
{I_2} &\le& \ln (2{s^2}) + \frac{{2\left| \mu  \right|}}{{{s^2}}}\sqrt {\frac{{\frac{{H_{\rm{E}}^2}}{{H_{\rm{B}}^2}}\sigma _{\rm{B}}^2 + \sigma _{\rm{E}}^2}}{{2\pi }}}  + \frac{{\left| {1 - \mu \frac{{{H_{\rm{E}}}}}{{{H_{\rm{B}}}}}} \right|}}{{{s^2}}}\nonumber\\
&\times&{E_{{X^*}}}\left[ {\int_{ - \infty }^\infty  {\frac{{{e^{ - \frac{{{n^2}}}{{2\sigma _{\rm{B}}^2}}}}}}{{\sqrt {2\pi } {\sigma _{\rm{B}}}}}\left| {n + {H_{\rm{B}}}X} \right|{\rm{d}}n} } \right].
 \label{eq62}
\end{eqnarray}
Because $\left| {a + b} \right| \le \left| a \right| + \left| b \right|$ always holds,
eq. (\ref{eq62}) can be further upper-bounded by
\begin{eqnarray}
\!\!\!\!\!\!\!\!\!\!\!\!{I_2} \!\!\!\!\!&\le&\!\!\!\!\! \ln (2{s^2}) + \frac{2}{{{s^2}}}\nonumber\\
\!\!\!\!\!\!\!\!\!\!\!\!&\times&\!\!\!\!\!\!\underbrace {\left[ {\left| \mu  \right|\!\sqrt {\!\frac{{\frac{{H_{\rm{E}}^2}}{{H_{\rm{B}}^2}}\sigma _{\rm{B}}^2 \!+\! \sigma _{\rm{E}}^2}}{{2\pi }}}  \!+\! \left| {1 \!-\! \mu \frac{{{H_{\rm{E}}}}}{{{H_{\rm{B}}}}}} \right|\!\left(\! {\frac{{{\sigma _{\rm{B}}}}}{{\sqrt {2\pi } }} \!+\! \frac{{{H_{\rm{B}}}\xi P}}{2}} \!\right)}\!\! \right]}_{{I_3}}\!\!.
 \label{eq63}
\end{eqnarray}
To obtain a tight upper bound of $I_2$, the minimum value of $I_3$ in (\ref{eq63}) should be determined.
In the following, three cases are considered:

Case 1: when $\mu  < 0$, $I_3$ is given by
\begin{eqnarray}
{I_3} \!\!\!\!\!&=&\!\!\!\!\! -\! \left[\!\! {\sqrt {\!\!\frac{{\frac{{H_{\rm{E}}^2}}{{H_{\rm{B}}^2}}\sigma _{\rm{B}}^2 \!\!+\!\! \sigma _{\rm{E}}^2}}{{2\pi }}} \!+\!\! \frac{{{H_{\rm{E}}}}}{{{H_{\rm{B}}}}}\!\!\left(\! {\frac{{{\sigma _{\rm{B}}}}}{{\sqrt {2\pi } }} \!\!+\!\! \frac{{{H_{\rm{B}}}\xi P}}{2}}\! \right)}\!\! \right]\!\mu \! +\! \frac{{{\sigma _{\rm{B}}}}}{{\sqrt {2\pi } }} \!+\! \frac{{{H_{\rm{B}}}\xi P}}{2} \nonumber \\
 &\ge&\!\!\!\!\! \frac{{{\sigma _{\rm{B}}}}}{{\sqrt {2\pi } }} + \frac{{{H_{\rm{B}}}\xi P}}{2};
 \label{eq64}
\end{eqnarray}

Case 2: when $0 \le \mu  \le {{{H_{\rm{B}}}} \mathord{\left/
 {\vphantom {{{H_{\rm{B}}}} {{H_{\rm{E}}}}}} \right.
 \kern-\nulldelimiterspace} {{H_{\rm{E}}}}}$, $I_3$ is given by
\begin{eqnarray}
{I_3} &=& \left[ {\sqrt {\frac{{\frac{{H_{\rm{E}}^2}}{{H_{\rm{B}}^2}}\sigma _{\rm{B}}^2 + \sigma _{\rm{E}}^2}}{{2\pi }}}  - \frac{{{H_{\rm{E}}}}}{{{H_{\rm{B}}}}}\left( {\frac{{{\sigma _{\rm{B}}}}}{{\sqrt {2\pi } }} + \frac{{{H_{\rm{B}}}\xi P}}{2}} \right)} \right]\mu \nonumber\\
 &+& \frac{{{\sigma _{\rm{B}}}}}{{\sqrt {2\pi } }} + \frac{{{H_{\rm{B}}}\xi P}}{2}.
\label{eq65}
\end{eqnarray}
If $\sqrt {( H_{\rm{E}}^2 \sigma _{\rm{B}}^2 / H_{\rm{B}}^2 + \sigma _{\rm{E}}^2 ) / (2\pi ) } \ge
H_{\rm{E}} ( \sigma _{\rm{B}} / {\sqrt{2\pi}} + H_{\rm{B}}\xi P/2 )/H_{\rm{B}}$,
we have ${I_3} \ge \sigma _{\rm{B}} / {\sqrt{2\pi}} + {H_{\rm{B}}}\xi P /2$.
Otherwise, we have ${I_3} \ge \sqrt {( {\sigma _{\rm{B}}^2 + H_{\rm{B}}^2\sigma _{\rm{E}}^2/H_{\rm{E}}^2} )/(2\pi )}$;

Case 3: when $\mu  > {{{H_{\rm{B}}}} \mathord{\left/
 {\vphantom {{{H_{\rm{B}}}} {{H_{\rm{E}}}}}} \right.
 \kern-\nulldelimiterspace} {{H_{\rm{E}}}}}$, $I_3$ is given by
\begin{eqnarray}
{I_3}\!\!\!\!\!\!& =&\!\!\!\!\!\!\!\left[\!\! {\sqrt {\!\!\frac{{\frac{{H_{\rm{E}}^2}}{{H_{\rm{B}}^2}}\sigma _{\rm{B}}^2 \!+\! \sigma _{\rm{E}}^2}}{{2\pi }}}  \!+\! \frac{{{H_{\rm{E}}}}}{{{H_{\rm{B}}}}}\!\!\left(\!\! {\frac{{{\sigma _{\rm{B}}}}}{{\sqrt {2\pi } }} \!+\! \frac{{{H_{\rm{B}}}\xi P}}{2}}\!\! \right)}\!\! \right]\!\mu  \!-\! \frac{{{\sigma _{\rm{B}}}}}{{\sqrt {2\pi } }} \!-\! \frac{{{H_{\rm{B}}}\xi P}}{2} \nonumber \\
 &>&\!\!\!\!\!\!\! \sqrt {\frac{{\sigma _{\rm{B}}^2 + \frac{{H_{\rm{B}}^2}}{{H_{\rm{E}}^2}}\sigma _{\rm{E}}^2}}{{2\pi }}}.
\label{eq66}
\end{eqnarray}

According to the above three cases, we have
\begin{equation}
{I_3} \!\ge\!\!\! \left\{ \begin{array}{l}\!\!\!\!
\frac{{{\sigma _{\rm{B}}}}}{{\sqrt {2\pi } }} \!\!+\!\! \frac{{{H_{\rm{B}}}\xi P}}{2},\;{\rm{if}}\;\sqrt {\!\!\frac{{\frac{{H_{\rm{E}}^2}}{{H_{\rm{B}}^2}}\sigma _{\rm{B}}^2 \!+\! \sigma _{\rm{E}}^2}}{{2\pi }}} \! \ge\! \frac{{{H_{\rm{E}}}}}{{{H_{\rm{B}}}}}\!\!\left(\!\! {\frac{{{\sigma _{\rm{B}}}}}{{\sqrt {2\pi } }} \!\!+\!\! \frac{{{H_{\rm{B}}}\xi P}}{2}}\!\! \right)\\
\!\!\!\!\sqrt {\!\!\frac{{\sigma _{\rm{B}}^2 \!+\! \frac{{H_{\rm{B}}^2}}{{H_{\rm{E}}^2}}\sigma _{\rm{E}}^2}}{{2\pi }}} ,\;{\rm{if}}\;\sqrt {\!\!\frac{{\frac{{H_{\rm{E}}^2}}{{H_{\rm{B}}^2}}\sigma _{\rm{B}}^2 \!+\! \sigma _{\rm{E}}^2}}{{2\pi }}}  \!\le\! \frac{{{H_{\rm{E}}}}}{{{H_{\rm{B}}}}}\!\!\left(\!\! {\frac{{{\sigma _{\rm{B}}}}}{{\sqrt {2\pi } }} \!\!+\!\! \frac{{{H_{\rm{B}}}\xi P}}{2}} \!\!\right)\!.
\end{array} \right.
\label{eq67}
\end{equation}

Substituting (\ref{eq67}) into (\ref{eq63}), $I_2$ is further upper-bounded by
\begin{equation}
{I_2} \!\!\le\!\!\! \left\{ \begin{array}{l}\!\!\!\!
\ln\! \left[ {4e\!\left(\! {\frac{{{\sigma _{\rm{B}}}}}{{\sqrt {2\pi } }} \!+\! \frac{{{H_{\rm{B}}}\xi P}}{2}} \!\right)}\! \right],{\kern 1pt} {\rm{if}}\;\sqrt {\!\!\frac{{\frac{{H_{\rm{E}}^2}}{{H_{\rm{B}}^2}}\sigma _{\rm{B}}^2 \!\!+\!\! \sigma _{\rm{E}}^2}}{{2\pi }}} \! \ge\! \frac{{{H_{\rm{E}}}}}{{{H_{\rm{B}}}}}\!\!\left(\! {\frac{{{\sigma _{\rm{B}}}}}{{\sqrt {2\pi } }} \!+\! \frac{{{H_{\rm{B}}}\xi P}}{2}}\! \right)\\
\!\!\!\!\ln\! \left(\!\! {4e\sqrt {\frac{{\sigma _{\rm{B}}^2 + \frac{{H_{\rm{B}}^2}}{{H_{\rm{E}}^2}}\sigma _{\rm{E}}^2}}{{2\pi }}} } \!\!\right),\;{\rm{if}}\;\sqrt {\!\!\frac{{\frac{{H_{\rm{E}}^2}}{{H_{\rm{B}}^2}}\sigma _{\rm{B}}^2 \!\!+\!\! \sigma _{\rm{E}}^2}}{{2\pi }}} \! \le\! \frac{{{H_{\rm{E}}}}}{{{H_{\rm{B}}}}}\!\!\left(\! {\frac{{{\sigma _{\rm{B}}}}}{{\sqrt {2\pi } }} \!+\! \frac{{{H_{\rm{B}}}\xi P}}{2}}\! \right).
\end{array} \right.
\label{eq68}
\end{equation}
Finally, submitting (\ref{eq58}) and (\ref{eq68}) into (\ref{eq51}), eq. (\ref{eq17}) can be derived.

\section{Proof of Asymptotic Behavior (\ref{eq17_1}) in Corollary \ref{cor1}}
\label{appc1}
\renewcommand{\theequation}{D.\arabic{equation}}
Here, the asymptotic expression of the lower bound on secrecy capacity is based on (\ref{eq7}) in \emph{Theorem \ref{them1}}.
To facilitate the analysis, eq. (\ref{eq7}) can be further written as
\begin{eqnarray}
C_{\rm s} \!\!\!\!\! &\ge&\!\!\!\!\! \underbrace{\frac{1}{2}\!\ln \!\left(\! {1 \!+\! \frac{{H_{\rm{B}}^2{\xi ^2}{P^2}e}}{{2\pi \sigma _{\rm{B}}^2}}}\! \right)}_{I_4} \!-\! \left\{\! {\ln \!\left[ {\beta {e^{ - \frac{{{\delta ^2}}}{{2\sigma _{\rm{E}}^2}}}} \!+\! \sqrt {2\pi } {\sigma _{\rm{E}}}{\cal Q}\left(\! {\frac{\delta }{{{\sigma _{\rm{E}}}}}} \!\right)}\! \right] } \right.\nonumber\\
&+&\!\!\!\!\! \frac{1}{2}{\cal Q}\left(\! {\frac{\delta }{{{\sigma _{\rm{E}}}}}} \!\right)\!+\! \frac{\delta }{{2\sqrt {2\pi } {\sigma _{\rm{E}}}}}{e^{ - \frac{{{\delta ^2}}}{{2\sigma _{\rm{E}}^2}}}} \!+\! \frac{{{\delta ^2}}}{{2\sigma _{\rm{E}}^2}}{\cal Q}\left(\! { - \frac{{\delta  \!+\! {H_{\rm{E}}}\xi P}}{{{\sigma _{\rm{E}}}}}}\! \right)\nonumber\\
&+&\!\!\!\!\!\left. {    \frac{{\delta  + {H_{\rm{E}}}\xi P}}{\beta } + \frac{{{\sigma _{\rm{E}}}}}{{\sqrt {2\pi } \beta }}{e^{ - \frac{{{\delta ^2}}}{{2\sigma _{\rm{E}}^2}}}} - \frac{1}{2}\ln \left( {2\pi e\sigma _{\rm{E}}^2} \right)} \right\}\nonumber\\
&\triangleq&\!\!\!\! I_4 -I_5.
\label{eqd1}
\end{eqnarray}
For $I_4$, we have
\begin{eqnarray}
\mathop {\lim }\limits_{P \to \infty } I_4   -\mathop {\lim }\limits_{P \to \infty }  \ln (H_{\rm E}\xi P) = \frac{1}{2}\ln \left( {\frac{{eH_{\rm{B}}^2}}{{2\pi H_{\rm E}^2\sigma _{\rm{B}}^2}}} \right).
\label{eqd2}
\end{eqnarray}
Note that selecting any $\delta  \ge 0$ and $\beta $ can result in a lower bound of (\ref{eqd1}).
Without loss of generality, we choose $\delta  = {\sigma _{\rm{E}}}\sqrt {\ln \left( {\frac{{{H_{\rm{E}}}\xi P}}{{{\sigma _{\rm{E}}}}}} \right)}$ and $\beta  = {H_{\rm{E}}}\xi P$. Therefore, we have
\begin{eqnarray}
\mathop {\lim }\limits_{P \to \infty }\!\! \left\{ {I_5 \!-\! \ln \left( {{H_{\rm{E}}}\xi P} \right)} \right\}\!= \frac{1}{2}\ln \left( {\frac{e}{{2\pi \sigma _{\rm{E}}^2}}} \right).
\label{eqd4}
\end{eqnarray}
By using (\ref{eqd2}) and (\ref{eqd4}), we have
\begin{eqnarray}
\mathop {\lim }\limits_{P \to \infty } C_{\rm s} \ge  \ln \left( {\frac{{{H_{\rm{B}}}{\sigma _{\rm{E}}}}}{{{H_{\rm{E}}}{\sigma _{\rm{B}}}}}} \right).
 \label{eqd7}
\end{eqnarray}

The asymptotic expression of the upper bound on secrecy capacity is based on (\ref{eq17}) in \emph{Theorem \ref{them2}}.
According to (\ref{eq17}), when $P \to \infty $, we have
\begin{eqnarray}
\mathop {\lim }\limits_{P \to \infty } C_{\rm s} \le \ln \left( {\frac{{2\sqrt e {H_{\rm{B}}}{\sigma _{\rm{E}}}}}{{\pi {H_{\rm{E}}}{\sigma _{\rm{B}}}}}} \right).
 \label{eqd8}
\end{eqnarray}
According to (\ref{eqd7}) and (\ref{eqd8}), eq. (\ref{eq17_1}) can be derived.

\section{Proof of lower bound (\ref{eq19}) in Theorem \ref{them3}}
\label{appd}
\renewcommand{\theequation}{E.\arabic{equation}}
According to (\ref{eq28}), the secrecy capacity in this case can also be written as
\begin{equation}
C_{\rm s} \ge {C_{\rm{B}}} - {C_{\rm{E}}}.
\label{eq69}
\end{equation}

When $\alpha  \in (0,0.5)$, ${C_{\rm{B}}}$ is lower-bounded by \cite{BIB22}
\begin{equation}
{C_{\rm{B}}} \ge \frac{1}{2}\ln \left[ {1 + H_{\rm{B}}^2{A^2}\frac{{{e^{2\alpha \tilde \mu }}}}{{2\pi e\sigma _{\rm{B}}^2}}{{\left( {\frac{{1 - {e^{ - \tilde \mu }}}}{{\tilde \mu }}} \right)}^2}} \right],
\label{eq70}
\end{equation}
where $\tilde \mu $ is the solution to (\ref{eq23}).
Moreover, ${C_{\rm{E}}}$ is upper-bounded by \cite{BIB22}
\begin{eqnarray}
{C_{\rm{E}}}\!\!\!\!\!&\le& \!\!\!\!\! {\cal Q}\left(\! {\frac{\delta }{{{\sigma _{\rm{E}}}}}}\! \right) \!+\! \frac{\delta }{{\sqrt {2\pi } {\sigma _{\rm{E}}}}}{e^{ - \frac{{{\delta ^2}}}{{2\sigma _{\rm{E}}^2}}}} \!-\! \frac{1}{2} \!+\!\mu \alpha \!\!\left[ {1 \!-\! 2{\cal Q}\!\left(\! {\frac{{\delta  \!+\! \frac{{{H_{\rm{E}}}A}}{2}}}{{{\sigma _{\rm{E}}}}}} \!\right)} \right]\nonumber \\
 &+& \!\!\!\! \left[ {{\cal Q}\left( { - \frac{{\delta  + {H_{\rm{E}}}\alpha A}}{{{\sigma _{\rm{E}}}}}} \right) - {\cal Q}\left( {\frac{{\delta  + (1 - \alpha ){H_{\rm{E}}}A}}{{{\sigma _{\rm{E}}}}}} \right)} \right]\nonumber\\
 &\times&\!\!\!\! \ln \left[ {\frac{{{H_{\rm{E}}}A}}{{\sqrt {2\pi } {\sigma _{\rm{E}}}\mu }}.\frac{{{e^{\frac{{\mu \delta }}{{{H_{\rm{E}}}A}}}} - {e^{ - \mu \left( {1 + \frac{\delta }{{{H_{\rm{E}}}A}}} \right)}}}}{{\left( {1 - 2{\cal Q}\left( {\frac{\delta }{{{\sigma _{\rm{E}}}}}} \right)} \right)}}} \right] \nonumber \\
 &+& \!\!\!\! \frac{{\mu {\sigma _{\rm{E}}}}}{{{H_{\rm{E}}}A\sqrt {2\pi } }}\!\!\left( {{e^{ - \frac{{{\delta ^2}}}{{2\sigma _{\rm{E}}^2}}}} \!-\! {e^{ - \frac{{{{({H_{\rm{E}}}A + \delta )}^2}}}{{2\sigma _{\rm{E}}^2}}}}} \right),
 \label{eq71}
\end{eqnarray}
where $\mu $ and $\delta $ is given by (\ref{eq24}).
Substituting (\ref{eq70}) and (\ref{eq71}) into (\ref{eq69}), ${C_{{\rm s},1}}$ is derived.

When $\alpha  \in [0.5,1]$, ${C_{\rm{B}}}$ is lower-bounded by \cite{BIB22}
\begin{equation}
{C_{\rm{B}}} \ge \frac{1}{2}\ln \left( {1 + \frac{{H_{\rm{B}}^2{A^2}}}{{2\pi e\sigma _{\rm{B}}^2}}} \right),
\label{eq72}
\end{equation}
and ${C_{\rm{E}}}$ is upper-bounded by \cite{BIB22}
\begin{eqnarray}
{C_{\rm{E}}} \!\!\!\!\!&\le&\!\!\!\!\! \left[ {1 \!-\! 2{\cal Q}\!\left(\! {\frac{{\delta  \!+\! \frac{{{H_{\rm{E}}}A}}{2}}}{{{\sigma _{\rm{E}}}}}} \!\right)} \right]\!\ln\! \left[ {\frac{{{H_{\rm{E}}}A + 2\delta }}{{\sqrt {2\pi } {\sigma _{\rm{E}}}\left( {1 \!-\! 2{\cal Q}\left(\! {\frac{\delta }{{{\sigma _{\rm{E}}}}}} \!\right)} \right)}}} \right]\nonumber\\
 &+&\!\!\!\!\! {\cal Q}\left(\! {\frac{\delta }{{{\sigma _{\rm{E}}}}}} \!\right) \!+\! \frac{\delta }{{\sqrt {2\pi } {\sigma _{\rm{E}}}}}{e^{ - \frac{{{\delta ^2}}}{{2\sigma _{\rm{E}}^2}}}} \!-\! \frac{1}{2}.
\label{eq73}
\end{eqnarray}
Substituting (\ref{eq72}) and (\ref{eq73}) into (\ref{eq69}), ${C_{{\rm s},2}}$ is derived.

\section{Proof of lower bound (\ref{eq20}) in Theorem \ref{them3}}
\label{appe}
\renewcommand{\theequation}{F.\arabic{equation}}
In this case, eq. (\ref{eq37}) can also be derived.
Similar to (\ref{eq39}), a lower bound on the secrecy capacity can be derived by solving the following problem
\begin{eqnarray}
&&\mathop {\min }\limits_{{f_X}(x)} {\cal J}\left[ {{f_X}(x)} \right]\triangleq \int_0^A {{f_X}(x)\ln} \left[ {{f_X}(x)} \right]{\rm{d}}x \nonumber \\
{\rm{s.t.}}&& \int_0^A {{f_X}(x){\rm d}x}  = 1\nonumber \\
&&\int_0^A {x{f_X}(x){\rm d}x}  = \xi P.
\label{eq74}
\end{eqnarray}
By employing the variational method\cite{BIB20_add3}, the input PDF is derived as
\begin{equation}
{f_X}(x) = {e^{cx + b - 1}},
\label{eq76}
\end{equation}
where $b$ and $c$ are two free parameters.

When $c=0$, submitting (\ref{eq76}) into the two constraints in (\ref{eq74}), we have
\begin{equation}
{f_X}(x) = \left\{ \begin{array}{l}
\frac{1}{A},\;x \in [0,A]\\
0,\;\;{\rm{otherwise}}
\end{array} \right.,
\label{eq77}
\end{equation}
and $A = 2\xi P$, i.e., $\alpha  = 0.5$.
Therefore, ${\cal H}(X)$ and ${\mathop{\rm var}} ({Y_{\rm{E}}})$ can be written, respectively, as
\begin{eqnarray}
\left\{ \begin{array}{l}
{\cal H}(X) = \ln (A) \\
{\mathop{\rm var}} ({Y_{\rm{E}}}) = H_{\rm{E}}^2\frac{{{\xi ^2}{P^2}}}{3} + \sigma _{\rm{E}}^2.
\end{array} \right.
\label{eq78}
\end{eqnarray}
Submitting (\ref{eq78}) into (\ref{eq37}), the lower bound on secrecy capacity for $\alpha  = 0.5$ is derived.

When $c \ne 0$, $\alpha  \ne 0.5$ and $\alpha  \in (0,1]$, submit (\ref{eq76}) into the two constraints in (\ref{eq74}), we have
\begin{equation}
{f_X}(x) = \left\{ \begin{array}{l}
\frac{{c{e^{cx}}}}{{{e^{cA}} - 1}},x \in [0,A]\\
0,\;{\rm{otherwise}}
\end{array} \right.,
\label{eq80}
\end{equation}
where $c$ is the solution to (\ref{eq25}).

Furthermore, ${\cal H}(X)$ and ${\mathop{\rm var}} ({Y_{\rm{E}}})$ can be written, respectively, as
\begin{eqnarray}
\left\{ \begin{array}{l}
{\cal H}(X) = \ln \left[ {{e^{ - c\xi P}}\left( {\frac{{{e^{cA}} - 1}}{c}} \right)} \right] \\
{\mathop{\rm var}} ({Y_{\rm{E}}}) = H_{\rm{E}}^2\left[ {\frac{{A(cA - 2)}}{{c(1 - {e^{ - cA}})}} + \frac{2}{{{c^2}}} - {\xi ^2}{P^2}} \right] + \sigma _{\rm{E}}^2.
\end{array} \right.
\label{eq81}
\end{eqnarray}
Submit (\ref{eq81}) into (\ref{eq37}),
the lower bound on secrecy capacity for $\alpha \! \ne\! 0.5$ and $\alpha  \!\in\! (0,1]$ is derived.

\section{Proof of upper bound (\ref{eq26}) in Theorem \ref{them3}}
\label{appf}
\renewcommand{\theequation}{G.\arabic{equation}}
From (\ref{eq51}), we have
\begin{equation}
C_{\rm s} \le {I_1} + {I_{\rm{2}}}.
\label{eq83}
\end{equation}

According to (\ref{eq52})-(\ref{eq58}), ${I_1}$ in this case can also be expressed as
\begin{equation}
{I_1} =  - \frac{1}{2}{\rm{ln}}\left[ {2\pi e\sigma _{\rm{B}}^2\left( {1 + \frac{{H_{\rm{E}}^2\sigma _{\rm{B}}^2}}{{H_{\rm{B}}^2\sigma _{\rm{E}}^2}}} \right)} \right].
\label{eq84}
\end{equation}

To obtain ${I_2}$, ${g_{{Y_{\rm{B}}}|{Y_E}}}({y_{\rm{B}}}|{y_{\rm{E}}})$ is chosen as
\begin{equation}
{g_{{Y_{\rm{B}}}|{Y_{\rm E}}}}({y_{\rm{B}}}|{y_{\rm{E}}}) = \frac{1}{{\sqrt {2\pi } s}}{e^{ - \frac{{{{({y_{\rm{B}}} - \mu {y_{\rm{E}}})}^2}}}{{2{s^2}}}}},
\label{eq85}
\end{equation}
where $\mu $ and $s$ are free parameters to be determined.

According to (\ref{eq85}) and (\ref{eq60}), ${I_2}$ can be derived as
\begin{eqnarray}
\!\!\!\!\!\!\!\!\!{I_2} \!\!\!\!\!&=&\!\!\!\!\!  \frac{1}{2}\ln (2\pi {s^2}) \nonumber\\
\!\!\!\!\!\!\!\!\!&+&\!\!\!\!\!\!{E_{{X^*}}}\!\!\!\left\{\!\! {\frac{{{{\left(\! {1 \!-\! \mu \frac{{{H_{\rm{E}}}}}{{{H_{\rm{B}}}}}} \!\right)\!}^2}(H_{\rm{B}}^2{X^2} \!+\! \sigma _{\rm{B}}^2) \!+\! {\mu ^2}\!\!\left(\! {\frac{{H_{\rm{E}}^2}}{{H_{\rm{B}}^2}}\sigma _{\rm{B}}^2 \!+\! \sigma _{\rm{E}}^2}\! \right)}}{{2{s^2}}}}\!\! \right\}\!\!.
 \label{eq86}
\end{eqnarray}
Owing to $0 \le X \le A$ and ${E_{{X^*}}}(X) = \xi P$, we have
\begin{eqnarray}
{E_{{X^*}}}({X^2})  \le \int_0^A {Ax{f_{{X^*}}}(x){\rm{d}}x} = A\xi P.
  \label{eq87}
\end{eqnarray}
Therefore, using (\ref{eq87}), eq. (\ref{eq86}) can be further written as
\begin{eqnarray}
{I_2} \!\!\!\!\!&\leq&\!\!\!\!\!  \frac{1}{2}\ln (2\pi {s^2}) \nonumber\\
&+&\!\!\!\!\! \frac{{{{\left(\! {1 \!-\! \mu \frac{{{H_{\rm{E}}}}}{{{H_{\rm{B}}}}}} \!\right)\!}^2}(H_{\rm{B}}^2A\xi P \!+\! \sigma _{\rm{B}}^2) \!+\! {\mu ^2}\!\!\left( {\frac{{H_{\rm{E}}^2}}{{H_{\rm{B}}^2}}\sigma _{\rm{B}}^2 \!+\! \sigma _{\rm{E}}^2} \right)}}{{2{s^2}}}.
  \label{eq88}
\end{eqnarray}
To maximize the term on the right hand side of (\ref{eq88}),
taking the first partial derivative with respect to $\mu $ and $s$, and letting them to be zero, we have
\begin{equation}
\left\{ \begin{array}{l}
\mu {\rm{ = }}\frac{{\frac{{{H_{\rm{E}}}}}{{{H_{\rm{B}}}}}\left( {H_{\rm{B}}^2A\xi P + \sigma _{\rm{B}}^{\rm{2}}} \right)}}{{H_{\rm{E}}^{\rm{2}}A\xi P + 2\frac{{H_{\rm{E}}^{\rm{2}}}}{{H_{\rm{B}}^{\rm{2}}}}\sigma _{\rm{B}}^{\rm{2}} + \sigma _{\rm{E}}^{\rm{2}}}}\\
{s^2} = \frac{{\left( {\frac{{H_{\rm{E}}^{\rm{2}}}}{{H_{\rm{B}}^{\rm{2}}}}\sigma _{\rm{B}}^{\rm{2}} + \sigma _{\rm{E}}^{\rm{2}}} \right)\left( {H_{\rm{B}}^2A\xi P + \sigma _{\rm{B}}^{\rm{2}}} \right)}}{{H_{\rm{E}}^2A\xi P + 2\frac{{H_{\rm{E}}^{\rm{2}}}}{{H_{\rm{B}}^{\rm{2}}}}\sigma _{\rm{B}}^{\rm{2}} + \sigma _{\rm{E}}^{\rm{2}}}}.
\end{array} \right.
\label{eq90}
\end{equation}
In this case, eq. (\ref{eq88}) can be written as
\begin{equation}
{I_2} \leq  \frac{1}{2}\ln \left[ {2\pi e\frac{{\left( {\frac{{H_{\rm{E}}^{\rm{2}}}}{{H_{\rm{B}}^{\rm{2}}}}\sigma _{\rm{B}}^{\rm{2}} + \sigma _{\rm{E}}^{\rm{2}}} \right)\left( {H_{\rm{B}}^2A\xi P + \sigma _{\rm{B}}^{\rm{2}}} \right)}}{{H_{\rm{E}}^2A\xi P + 2\frac{{H_{\rm{E}}^{\rm{2}}}}{{H_{\rm{B}}^{\rm{2}}}}\sigma _{\rm{B}}^{\rm{2}} + \sigma _{\rm{E}}^{\rm{2}}}}} \right].
\label{eq91}
\end{equation}
Substituting (\ref{eq84}) and (\ref{eq91}) into (\ref{eq83}), eq. (\ref{eq26}) can be derived.

\section{Proof of Asymptotic Behavior (\ref{eq26_1}) in Corollary \ref{cor2}}
\label{appg}
\renewcommand{\theequation}{H.\arabic{equation}}
In this section, the asymptotic expression of the lower bound on secrecy capacity is based on (\ref{eq19}) in \emph{Theorem \ref{them3}}.

When $0 < \alpha  < 0.5$ in (\ref{eq19}), we have
\begin{eqnarray}
\mathop {\lim }\limits_{A \to \infty } {C_{{\rm s},1}} \!\!\!\!&=&\!\!\!\! \underbrace {\mathop {\lim }\limits_{A \to \infty } \frac{1}{2}\ln \left[ {1 + H_{\rm{B}}^2{A^2}\frac{{{e^{2\alpha \tilde \mu }}{{(1 - {e^{ - \tilde \mu }})}^2}}}{{2\pi e\sigma _{\rm{B}}^2{{\tilde \mu }^2}}}} \right]}_{I_6} \nonumber\\
&-& \mathop {\lim }\limits_{A \to \infty } \left\{ {{\cal Q}\left( {\frac{\delta }{{{\sigma _{\rm{E}}}}}} \right) + \frac{\delta }{{\sqrt {2\pi } {\sigma _{\rm{E}}}}}{e^{ - \frac{{{\delta ^2}}}{{2\sigma _{\rm{E}}^2}}}} - \frac{1}{2}} \right.\nonumber \\
 &+&\!\!\!\! \left[\! {{\cal Q}\left( \!{ - \frac{{\delta  \!+\! {H_{\rm{E}}}\alpha A}}{{{\sigma _{\rm{E}}}}}} \!\right) \!-\! {\cal Q}\left(\! {\frac{{\delta \! +\! (1 \!-\! \alpha ){H_{\rm{E}}}A}}{{{\sigma _{\rm{E}}}}}}\! \right)} \!\right]\nonumber\\
 &\times&\ln \left[ {\frac{{{H_{\rm{E}}}A}}{{\sqrt {2\pi } {\sigma _{\rm{E}}}\mu }}.\frac{{{e^{\frac{{\mu \delta }}{{{H_{\rm{E}}}A}}}} - {e^{ - \mu \left( {1 + \frac{\delta }{{{H_{\rm{E}}}A}}} \right)}}}}{{\left( {1 - 2{\cal Q}\left( {\frac{\delta }{{{\sigma _{\rm{E}}}}}} \right)} \right)}}} \right]\nonumber\\
 &+& \mu \alpha \left( {1 - 2{\cal Q}\left( {\frac{{\delta  + \frac{{{H_{\rm{E}}}A}}{2}}}{{{\sigma _{\rm{E}}}}}} \right)} \right) \nonumber\\
& +& \left. { \frac{{\mu {\sigma _{\rm{E}}}}}{{{H_{\rm{E}}}A\sqrt {2\pi } }}\left( {{e^{ - \frac{{{\delta ^2}}}{{2\sigma _{\rm{E}}^2}}}} - {e^{ - \frac{{{{({H_{\rm{E}}}A + \delta )}^2}}}{{2\sigma _{\rm{E}}^2}}}}} \right)} \right\}\nonumber\\
&\triangleq& I_6-I_7,
\label{eqh1}
\end{eqnarray}
where $\tilde \mu $ is the solution of (\ref{eq23}), $\mu $ and $\delta $ are given by (\ref{eq24}). Moreover, we have
\begin{eqnarray}
\mathop {\lim }\limits_{A \to \infty } I_6 - \mathop {\lim }\limits_{A \to \infty } \ln A \!\!\!&=&\!\!\! \frac{1}{2}\ln \left( {\frac{{H_{\rm{B}}^2}}{{2\pi e\sigma _{\rm{B}}^2}}} \right) \nonumber\\
&-&\!\!\! (1 - \alpha )\tilde \mu  - \ln ({\rm{1}} - \alpha \tilde \mu ).
\label{eqh4}
\end{eqnarray}
Furthermore, we have
\begin{eqnarray}
\mathop {\lim }\limits_{A \to \infty } I_7 \!-\! \mathop {\lim }\limits_{A \to \infty } \ln A \!\!\!\!&=&\!\!\!\!\ln \left( {\frac{{{H_{\rm{E}}}}}{{{\sigma _{\rm{E}}}}}} \right)-\frac{1}{2}\ln \left( {2\pi e} \right) \nonumber\\
&-&\!\!\!\!  (1 \!-\! \alpha )\tilde \mu  \!-\! \ln (1 \!-\! \alpha \tilde \mu ).
\label{eqh5}
\end{eqnarray}
From (\ref{eqh1}) and (\ref{eqh5}), we can get
\begin{eqnarray}
\mathop {\lim }\limits_{A \to \infty } C_{\rm s} &\ge& \ln \left( {\frac{{{H_{\rm{B}}}{\sigma _{\rm{E}}}}}{{{H_{\rm{E}}}{\sigma _{\rm{B}}}}}} \right).
 \label{eqh6}
\end{eqnarray}

When $0.5 \le \alpha  < 1$ in (\ref{eq19}), we have
\begin{eqnarray}
&&\!\!\!\!\!\!\!\!\!\mathop {\lim }\limits_{A \to \infty } {C_{{\rm s},2}}= \underbrace{\mathop {\lim }\limits_{A \to \infty } \frac{1}{2}\ln \!\!\left( {1 \!+\! \frac{{H_{\rm{B}}^2{A^2}}}{{2\pi e\sigma _{\rm{B}}^2}}} \right)}_{I_8} \nonumber\\
&&- \mathop {\lim }\limits_{A \to \infty }\!\! \left\{\!\! {\left[\! {1 \!-\! 2{\cal Q}\!\left(\!\! {\frac{{\delta \! +\! \frac{{{H_{\rm{E}}}A}}{2}}}{{{\sigma _{\rm{E}}}}}}\! \right)} \!\!\right]\!\ln \!\!\left[\!\! {\frac{{{H_{\rm{E}}}A \!+\! 2\delta }}{{\sqrt {2\pi } {\sigma _{\rm{E}}}\!\left(\! {1 \!-\! 2{\cal Q}\!\left(\! {\frac{\delta }{{{\sigma _{\rm{E}}}}}} \!\right)} \!\right)}}}\!\! \right]} \right.\nonumber \\
&&\left. { + {\cal Q}\left( {\frac{\delta }{{{\sigma _{\rm{E}}}}}} \right) + \frac{\delta }{{\sqrt {2\pi } {\sigma _{\rm{E}}}}}{e^{ - \frac{{{\delta ^2}}}{{2\sigma _{\rm{E}}^2}}}} - \frac{1}{2}} \right\}\nonumber\\
&&\triangleq I_8-I_9,
 \label{eqh7}
\end{eqnarray}
where $\delta $ is given by (\ref{eq24}). Moreover, we have
\begin{eqnarray}
\mathop {\lim }\limits_{A \to \infty } I_8 - \mathop {\lim }\limits_{A \to \infty } \ln A
 = \frac{1}{2}\ln \left( {\frac{{H_{\rm{B}}^2}}{{2\pi e\sigma _{\rm{B}}^2}}} \right).
  \label{eqh8}
\end{eqnarray}
Furthermore, we can obtain
\begin{eqnarray}
\mathop {\lim }\limits_{A \to \infty } I_9 - \mathop {\lim }\limits_{A \to \infty } \ln A
= \ln \left( {\frac{{{H_{\rm{E}}}}}{{\sqrt {2\pi } {\sigma _{\rm{E}}}}}} \right).
\label{eqh9}
\end{eqnarray}
According to (\ref{eqh7}), (\ref{eqh8}) and (\ref{eqh9}), eq. (\ref{eqh6}) can also be derived. This indicates that for all $\alpha  \in (0,1)$,
eq. (\ref{eqh6}) always holds.

In this section, the asymptotic expression of the upper bound on secrecy capacity is based on (\ref{eq26}) in \emph{Theorem \ref{them4}}.
According to (\ref{eq26}), when $A \to \infty $, we have
\begin{equation}
\mathop {\lim }\limits_{A \to \infty } C_{\rm s} \le \ln \left( {\frac{{{H_{\rm{B}}}{\sigma _{\rm{E}}}}}{{{H_{\rm{E}}}{\sigma _{\rm{B}}}}}} \right).
\label{eqh12}
\end{equation}
From (\ref{eqh6}) and (\ref{eqh12}), eq. (\ref{eq26_1}) can be derived.

\begin{IEEEbiography}[{\includegraphics[width=1in,height=1.25in,clip,keepaspectratio]{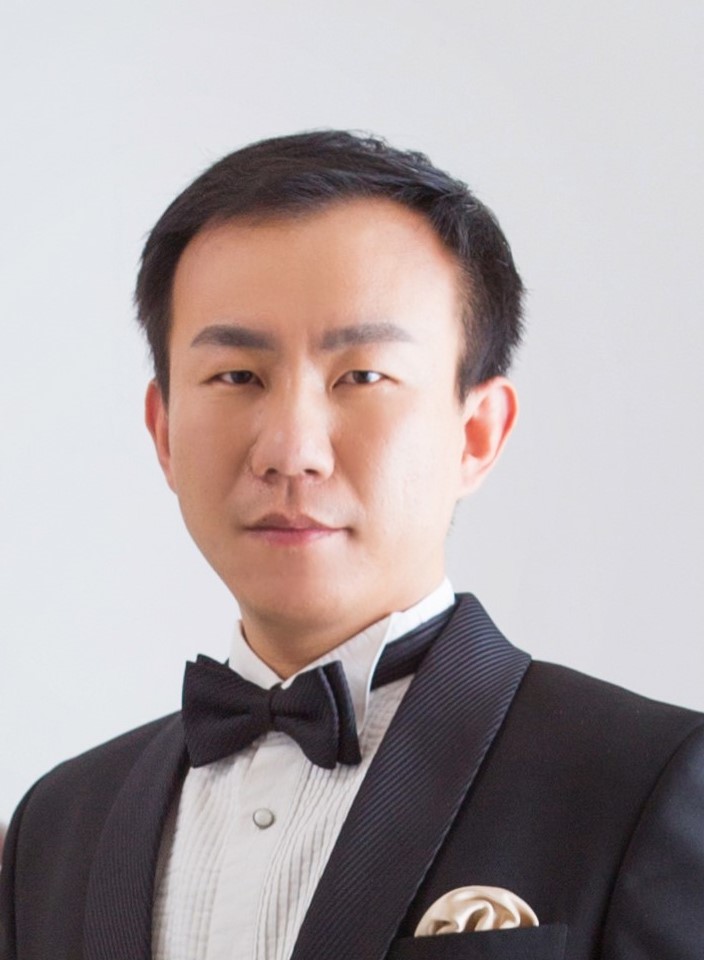}}]
{Jin-Yuan Wang} (S'12--M'16) received the B.S. degree in communication engineering from the College of Information and Electrical Engineering, Shandong University of Science and Technology, Qingdao, China, in 2009, the M.S. degree in electronic and communication engineering from the College of Electronic and Information Engineering, Nanjing University of Aeronautics and Astronautics, Nanjing, China, in 2012, and the PhD degree in Information and Communication Engineering from the National Mobile Communications Research Laboratory, Southeast University, Nanjing, China, in 2015.  He is currently a lecturer at College of Telecommunications and Information Engineering, Nanjing University of Posts and Telecommunications, Nanjing, China. His current research interest is visible light communications. He has authored/coauthored over 80 journal/conference papers. He has been a Technical Program Committee member for many international conferences, such as IEEE ICC and WTS. He also serves as a reviewer for many journals.
\end{IEEEbiography}

\begin{IEEEbiography}[{\includegraphics[width=1in,height=1.25in,clip,keepaspectratio]{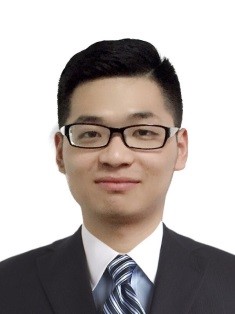}}]
{Cheng Liu} received the B.S. degree in communication engineering from the School of Information and Electrical Engineering, Hunan University of Science and Technology, China, in 2016.
He is currently studying for the M.S. degree in communication and information system from the National Mobile Communications Research Laboratory, Southeast University, Nanjing, China. His current research interest is physical-layer security in visible light communications.
\end{IEEEbiography}

\begin{IEEEbiography}[{\includegraphics[width=1in,height=1.25in,clip,keepaspectratio]{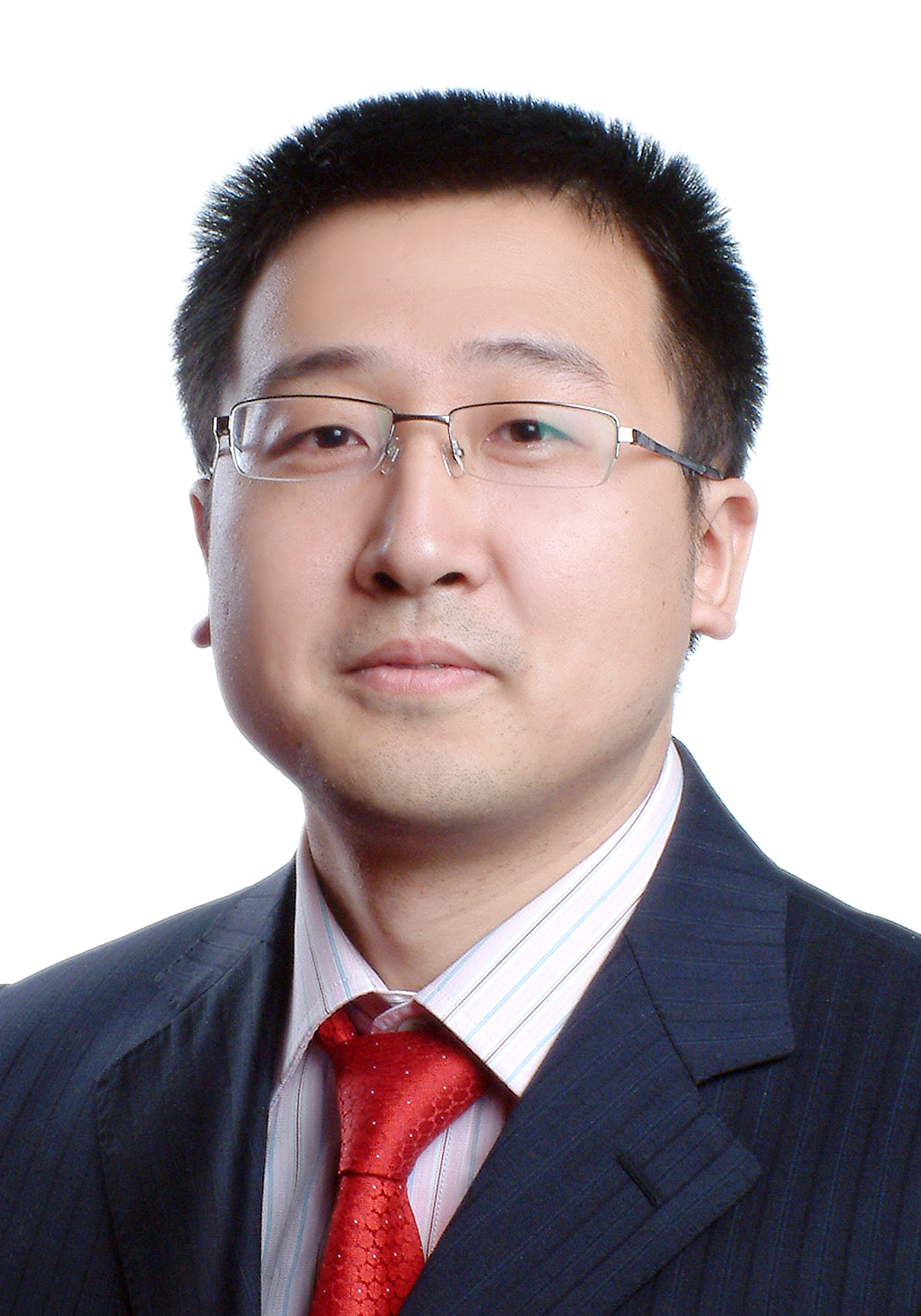}}]
{Jun-Bo Wang} (M'11) received the B.S. degree in computer science from the Hefei University of Technology, Hefei, China, in 2003, and the Ph.D. degree in communications engineering from the National Mobile Communications Research Laboratory, Southeast University, Nanjing, China, in 2008. He is currently an Associate Professor at National Mobile Communications Research Laboratory, Southeast University. From October of 2008 to August of 2013, he was with the Nanjing University of Aeronautics and Astronautics, China. From March of 2011 to February of 2013, he was a Postdoctoral Fellow at the National Laboratory for Information Science and Technology, Tsinghua University, Beijing, China. From September of 2016 to now, he is a Marie Sk\l odowska-Curie visiting scholar at the University of Kent, Kent, United Kingdom. His current research interests are wireless communications, signal processing, information theory and coding.
\end{IEEEbiography}

\begin{IEEEbiography}[{\includegraphics[width=1in,height=1.25in,clip,keepaspectratio]{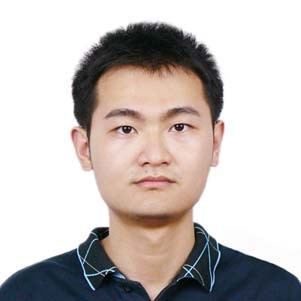}}]
{Yongpeng Wu} (S'08--M'13--SM'17) received the B.S. degree in telecommunication engineering from Wuhan University, Wuhan, China, in July 2007, the Ph.D. degree in communication and signal processing with the National Mobile Communications Research Laboratory, Southeast University, Nanjing, China, in November 2013. Dr. Wu is currently a senior research fellow with Institute for Communications Engineering, Technical University of Munich, Germany. Previously, he was the Humboldt research fellow and the senior research fellow with Institute for Digital Communications, University Erlangen-Nurnberg, Germany. During his doctoral studies, he conducted cooperative research at the Department of Electrical Engineering, Missouri University of Science and Technology, USA. His research interests include massive MIMO/MIMO systems, physical layer security, signal processing for wireless communications, and multivariate statistical theory. Dr. Wu was awarded the IEEE Student Travel Grants for IEEE International Conference on Communications (ICC) 2010, the Alexander von Humboldt Fellowship in 2014, the Travel Grants for IEEE Communication Theory Workshop 2016, and the Excellent Doctoral Thesis Awards of China Communications Society 2016. He was an Exemplary Reviewer of the IEEE Transactions on Communications in 2015, 2016. He is the lead guest editor for the special issue "Physical Layer Security for 5G Wireless Networks" of the IEEE Journal on Selected Areas in Communications. He is currently an editor of the IEEE Communications Letters. He has been a TPC member of various conferences, including Globecom, ICC, VTC, and PIMRC, etc.
\end{IEEEbiography}

\begin{IEEEbiography}[{\includegraphics[width=1in,height=1.25in,clip,keepaspectratio]{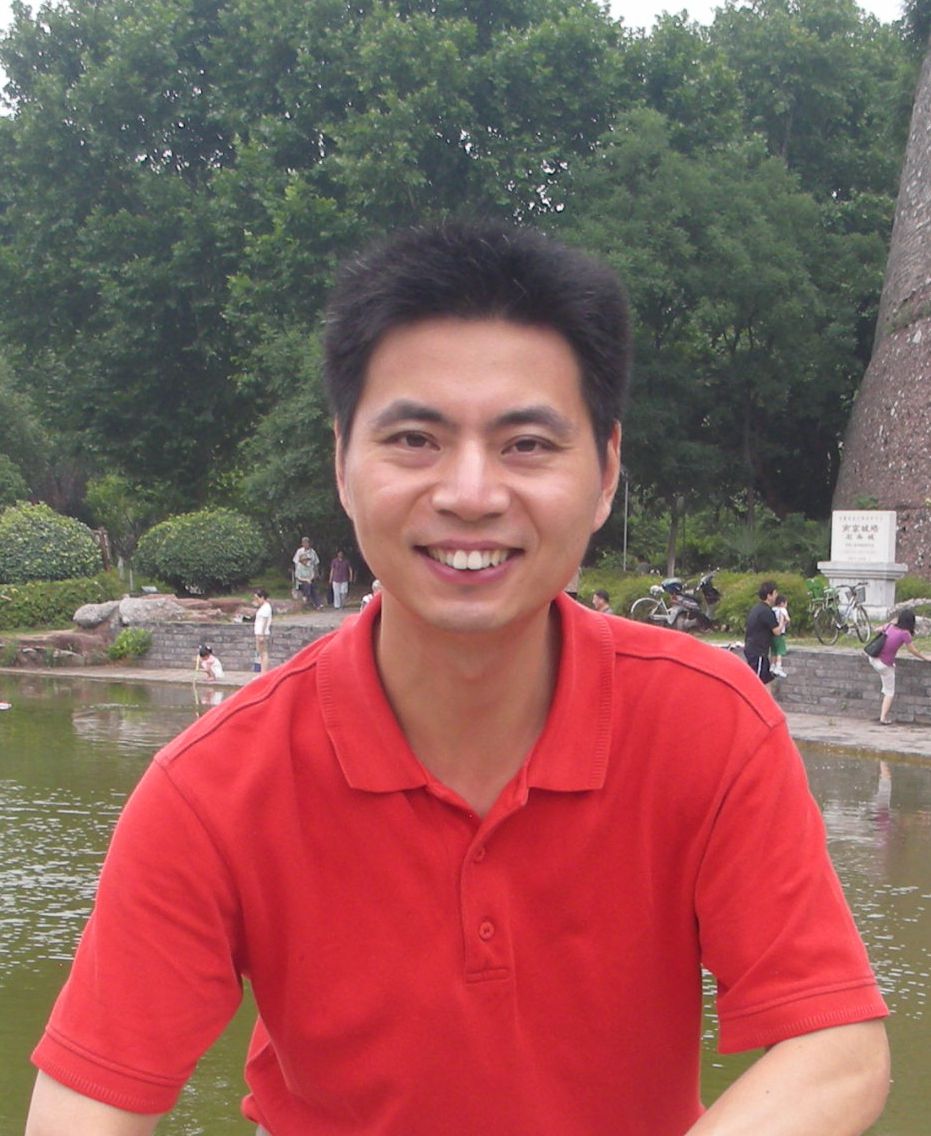}}]
{Min Lin} (M'13) received the B.S. degree from National University of Defense Technology, Changsa, China, in 1993, the M.S. degree from Nanjing Institute of Communication Engineering, Nanjing, China, in 2000, and the Ph.D. degree from Southeast University, Nanjing, in 2008, all in electrical engineering. He is currently a Professor with Nanjing University of Posts and Telecommunications. He has authored or coauthored over 100 papers. His current research interests include wireless communications and array signal processing. Dr. Lin has served as the TPC of many IEEE sponsored conferences, including IEEE ICC, Globecom, etc.
\end{IEEEbiography}

\begin{IEEEbiography}[{\includegraphics[width=1in,height=1.25in,clip,keepaspectratio]{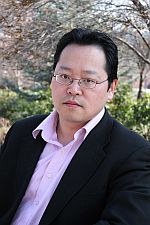}}]
{Julian Cheng} (S'96--M'04--SM'13) received the B.Eng. degree (Hons.) in electrical engineering from the University of Victoria, Victoria, BC, Canada, in 1995, the M.Sc. (Eng.) degree in mathematics and engineering from Queens University, Kingston, ON, Canada, in 1997, and the Ph.D. degree in electrical engineering from the University of Alberta, Edmonton, AB, Canada, in 2003. He is currently a Full Professor in the School of Engineering, Faculty of Applied Science, The University of British Columbia, Kelowna, BC, Canada. He was with Bell Northern Research and NORTEL Networks. His current research interests include digital communications over fading channels, statistical signal processing for wireless applications, optical wireless communications, and 5G wireless networks. He was the Co-Chair of the 12th Canadian Workshop on Information Theory in 2011, the 28th Biennial Symposium on Communications in 2016, and the 6th EAI International Conference on Game Theory for Networks (GameNets 216). He currently serves as an Area Editor for the IEEE TRANSACTIONS ON COMMUNICATIONS, and he was a past Associate Editor of the IEEE TRSACTIONS ON COMMUNICATIONS, the IEEE TRANSACTIONS ON WIRELESS COMMUNICATIONS, the IEEE COMMUNICATIONS LETTERS, and the IEEE ACCESS. Dr. Cheng served as a Guest Editor for a Special Issue of the IEEE JOURNAL ON SELECTED AREAS IN COMMUNICATIONS on Optical Wireless Communications. He is also a Registered Professional Engineer with the Province of British Columbia, Canada. Currently he serves as the President of the Canadian Society of Information Theory.
\end{IEEEbiography}

\end{document}